\documentclass[11pt]{article}
\usepackage{float}
\usepackage{booktabs}
\usepackage{latexsym}
\usepackage{enumerate}

\usepackage[utf8]{inputenc}
\usepackage[margin=1in]{geometry}
\usepackage[ruled,vlined]{algorithm2e}
\SetAlCapHSkip{0pt}       
\SetAlgoNlRelativeSize{0} 
\SetNlSty{}{}{}           

\usepackage{natbib}
\usepackage[english]{babel}  
\usepackage[T1]{fontenc}
\usepackage[utf8]{inputenc} 
\DeclareUnicodeCharacter{00A0}{~}
  
\usepackage{amssymb,amsmath,amscd,amsfonts,amsthm,bbm,mathrsfs,yhmath}
\usepackage{hyperref}
\usepackage{url}

\usepackage{tikz}
\usetikzlibrary{arrows.meta, positioning}
\tikzset{
    node/.style={circle, draw, minimum size=7mm, inner sep=0pt},
    blueNode/.style={circle, draw, blue, minimum size=7mm, inner sep=0pt},
    arrow/.style={-Stealth, thick},
    dashedArrow/.style={-Stealth, dashed, thick, gray},
}

\usepackage{xr-hyper}
\usepackage{algpseudocode}
\usepackage{setspace}
\usepackage{bm}

\theoremstyle{plain}
\newtheorem{corollary}{Corollary}
\newtheorem{theorem}{Theorem}
\newtheorem{lemma}{Lemma}

\newtheorem{assumption}{Assumption}

\newtheorem{example}{Example}

\usepackage{multirow}
\usepackage{tikz}
\usepackage{setspace}
\usetikzlibrary{shapes,decorations,arrows,calc,arrows.meta,fit,positioning}
\tikzset{
    -Latex,auto,node distance =1.5 cm and 1.5 cm,semithick,
    state/.style ={ellipse, draw, minimum width = 0.7 cm},
    point/.style = {circle, draw, inner sep=0.04cm,fill,node contents={}},
    directed/.style={Latex-Latex,dashed},
    el/.style = {inner sep=2pt, align=left, sloped}
}
\usepackage{bbding}
\usepackage{ulem}

\def\E{{\mathbb {E}}}

\def\N{{\mathcal{N}}}

\def\N{{\mathcal{N}}}

\def\bx{\bm{x}}
\def\bX{\bm{X}}

\def\bW{\bm{W}}
\def\bpi{\boldsymbol{\pi}}

\newcommand{\R}{{{\mathbb R}}}

\newcommand{\zeroitem}{\setlength{\itemsep}{0pt}
\setlength{\parsep}{0pt}
\setlength{\parskip}{0pt}}

\newcommand{\metahunt}{\textsc{MetaHunt}}
 
\DeclareMathOperator*{\argmax}{arg\,max}
\DeclareMathOperator*{\argmin}{arg\,min}

\def\khat{{\hat{K}}}

\title{\bf Privacy-preserving Meta-analysis through\\ Low-Rank Basis Hunting\thanks{An open-source software package \texttt{MetaHunt} is available at CRAN \url{https://cran.r-project.org/package=MetaHunt}. We thank Larry Han and an anonymous reviewer of IQSS's rapidPeer program for their helpful comments.}}
\author{Wenqi Shi\thanks{Department of Statistics, Harvard University, 1 Oxford Street, Cambridge, MA 02138, U.S.A. Email: \href{mailto:wenqi_shi@g.harvard.edu}{wenqi\_shi@g.harvard.edu}} \and Kosuke Imai\thanks{Department of Government and Department of Statistics, Harvard University.  1737 Cambridge Street, Institute for Quantitative Social Science, Cambridge MA 02138, U.S.A.  Email: \href{mailto:imai@harvard.edu}{imai@harvard.edu} URL: \href{https://imai.fas.harvard.edu}{https://imai.fas.harvard.edu}} \and Yi Zhang\thanks{Independent Researcher. Email: \href{mailto:yizhang0017@gmail.com}{yizhang0017@gmail.com}}}
\date{\today}

\begin{document}
\maketitle

\begin{abstract}
A central challenge of meta-analysis is that the populations underlying existing studies often differ from the target population in unknown ways. We study the problem of predicting function-valued quantities, such as regression and conditional average treatment effect functions, for a new target population using only study-level covariates and study-specific function estimates. We propose \metahunt, a new meta-analysis methodology based on a shared low-rank structure, in which the true function from each study lies within the convex hull of a small set of latent basis functions. To recover these basis functions, we extend a vertex-hunting procedure called the Successive Projection Algorithm to the functional setting, incorporating a denoised basis-hunting step. We establish consistency of the recovered basis functions under mild regularity conditions. We then model the relationship between study-level covariates and the corresponding mixing weights using flexible semi-parametric or non-parametric methods. \metahunt{} is privacy-preserving and enables meta-analytic prediction based on study-level information alone, even when individual-level data are unavailable to analysts. In addition, for each study, functions of interest can be estimated using possibly different machine learning algorithms. For uncertainty quantification, we construct prediction intervals via conformal prediction.  We show that, under exchangeability and mild estimation-error conditions, these intervals achieve asymptotically valid marginal coverage. We demonstrate the effectiveness of \metahunt{} through both simulation studies and empirical applications. An open-source software package is available for implementing \metahunt.

\end{abstract}

\begin{center}
\noindent Keywords: 
{conformal prediction; evidence synthesis; external validity; generalizability; \\ vertex hunting}
\end{center}

\clearpage
\onehalfspacing

\section{Introduction}
Modern scientific research is increasingly characterized by studies conducted across diverse populations, contexts, and designs. Integrating information from such heterogeneous sources is essential for assessing the external validity of empirical evidence. Meta-analysis provides a foundational framework for synthesizing evidence across studies and has been widely applied in a variety of fields including medicine, social sciences, and public policy \citep[e.g.,][]{hedgesolkin1985,normand1999meta,haidich2010meta,card2015applied,cooper2019handbook,tipton2019history}.

However, generalizing results from multiple source studies to a new target population is challenging, especially when there exists substantial heterogeneity both across sources and between the source and target populations. This heterogeneity extends beyond simple covariate shift to include conditional shifts, where outcome–covariate relationships vary across studies in unknown ways. While classical meta-analytic methods typically focus on scalar estimands such as average treatment effects, many contemporary applications require learning about richer, function-valued objects, including regression and conditional average treatment effect (CATE) functions.  Moreover, once these functional quantities are estimated, they can be used to infer scalar estimands.

In this paper, we consider the problem of predicting a function-valued quantity for a new target population using only aggregate-level data from multiple source studies. In our setting, each source study provides an estimated function of interest together with study-level covariates that summarize contextual features such as location, timing, and other aspects of study design. Crucially, we do not assume access to individual-level covariates or outcomes, reflecting common practical constraints in multi-site studies. Our goal is to predict the corresponding function for a new target population that is also characterized by the same set of study-level covariates. We seek a structured but flexible approach to modeling cross-study heterogeneity that enables principled generalization using only aggregate-level information.

We develop \metahunt, a unified estimation and inference framework for predicting function-valued quantities across heterogeneous studies. \metahunt{} is built on a low-rank structural assumption, which posits that all study-level functions lie in the convex hull of a small set of latent basis functions. This assumption captures cross-study heterogeneity through a low-dimensional representation and enables information sharing across studies without imposing parametric restrictions on the functions themselves. The intuition is that heterogeneity across studies is driven by a small number of persistent mechanisms. Study-level covariates then inform how these latent bases are combined differently across studies and in the target population.

To estimate the latent basis functions, we extend the Successive Projection Algorithm (SPA) \citep{araujo2001successive} from the vertex-hunting to functional setting.  The proposed functional SPA (fSPA) algorithm adapts a denoising step from \cite{jin2024improved} to improve robustness, along with data-driven procedures for selecting the number of basis functions. 
Given the bases recovered by this denoised fSPA (d-fSPA) algorithm, we estimate study-specific mixing weights by projecting each observed function onto the resulting convex hull. 

Lastly, we model the relationship between these weights and study-level covariates using flexible semi-parametric or nonparametric methods, allowing covariate information from the target population to inform prediction. The combination of recovered basis functions and the fitted weight model yields a predicted target function. To quantify uncertainty, we construct conformal prediction intervals \citep{vovk2005algorithmic, shafer2008tutorial}, which accommodate the multi-stage estimation pipeline without requiring distributional assumptions. All together, \metahunt{} is privacy-preserving in the sense that it requires no individual-level data from any of the source studies, while remaining compatible with modern machine learning methods for estimating study-level functions and fitting weight models.

Our theoretical analysis establishes consistency guarantees for the recovered basis functions under mild regularity conditions. To our knowledge, this work is the first to extend vertex-hunting algorithms to the functional setting. For uncertainty quantification, we show that the proposed conformal prediction intervals achieve distribution-free asymptotic marginal coverage under exchangeability and mild control of estimation error, despite the complexity of the underlying estimation procedure. Through extensive simulation studies and a real-world empirical application, we demonstrate that \metahunt{} provides robust prediction and valid uncertainty quantification. Finally, we make available an open-source software package \texttt{MetaHunt} at CRAN (\url{https://cran.r-project.org/package=MetaHunt}) that implements the proposed methodology.

\paragraph{Related Literature}

~\\
\noindent
Classical meta-analytic approaches typically focus on scalar estimands, most notably the average treatment effects (ATE). For function-valued quantities of interest, including the CATE, many widely used methods are based on parametric models such as mixed-effects or hierarchical models \citep[e.g.,][]{hedgesolkin1985, higgins2009re, lunn2013fully, debray2015get, burke2017meta, borenstein2021introduction, seo2021comparing, stogiannis2024heterogeneity}. While these hierarchical formulations allow CATE functions to vary across studies through parametric random effects, their primary estimands remain scalar ATEs. Moreover, such approaches require access to individual-level data and rely on strong parametric assumptions. When cross-study heterogeneity is complex or nonlinear, these methods are susceptible to bias due to model misspecification.  In addition, while aggregate-level meta-analytic approaches, including meta-regression, do not require individual-level information \citep[e.g.,][]{baker2009understanding, tipton2019history}, they are typically limited to scalar quantities and do not directly accommodate function-valued estimands.  In contrast, \metahunt{} targets function-valued quantities and allows the use of any flexible machine learning algorithm to estimate study-specific functions.

More recent work on causally interpretable meta-analysis emphasizes principled generalization to a well-defined target population, but often restricts attention to covariate shift and assumes transportability, namely the invariance of CATE functions across studies \citep[e.g.,][]{dahabreh2020toward,  clark2023causally, dahabreh2023efficient, rott2024causally, steingrimsson2024systematically, vo2025integration, cao2025heterogeneity, shi2025use}. In practice, however, transportability is often not credible because of various unobserved sources of cross-study heterogeneity. \metahunt{} departs from these approaches by allowing CATE functions to vary across studies in unknown ways without imposing parametric restrictions on the functions themselves. 

Another stream of literature that aggregates evidence from multiple datasets is federated learning. Federated learning, introduced by \citet{mcmahan2017communication}, emphasizes privacy-preserving collaborative model training, in which multiple decentralized sites (e.g., devices or institutions) jointly learn a model by sharing only summary-level information rather than individual-level data \citep{li2020federated, kairouz2021advances}. This paradigm has been increasingly applied in healthcare informatics, finance and industrial systems \citep[e.g.][]{silva2019federated, rieke2020future, long2020federated, li2020federated, xu2021federated, antunes2022federated}.
\metahunt{} offers a new aggregation method by modeling each study-level function as a convex combination of latent basis functions. While most existing work on federated causal inference relies on propensity-score-based adjustments and focuses primarily on covariate shift \citep[e.g.,][]{han2023multiply,xiong2023federated,han2025federated}, \metahunt{} allows the conditional relationship between treatment effects and covariates to vary flexibly across studies.

Our low-rank assumption is related to the distributionally robust optimization (DRO) literature, which studies decision-making under distributional uncertainty by optimizing worst-case performance over an uncertainty set \citep[e.g.,][]{ben2013robust, duchi2021statistics,zhang2024optimal}. Recent meta-analytic work models the target population as a convex combination of source-study distributions \citep[e.g.,][]{xiong2023distributionally, zhang2024minimax, wang2025distributionally, mo2024minimax}. Although this formulation is structurally similar to our low-rank assumption, the objectives and modeling choices differ in important ways. DRO methods focus on worst-case guarantees and are therefore driven primarily by source studies with extreme results. In particular, studies whose results lie within the convex hull of results from other source studies do not affect the solution. 
While this worst-case perspective offers strong robustness guarantees, it can be overly conservative in many practical applications. In contrast, \metahunt{} models the distribution of mixing weights, allowing all source studies to contribute meaningfully through their frequency and covariate patterns. 
Relatedly, \cite{jeong2024out} studies out-of-distribution prediction by approximating the target distribution as a weighted combination of source distributions through a random perturbation model. While their model is defined on the unconditional joint distribution, our low-rank assumption restricts only the conditional relationships between outcomes and predictors.


To the best of our knowledge, this work is the first to extend vertex-hunting algorithms to the functional setting. Vertex hunting has been studied extensively in finite-dimensional settings and arises in a wide range of applications, including hyperspectral unmixing \citep[e.g.,][]{bioucas2012hyperspectral}, archetypal analysis \citep[e.g.,][]{cutler1994archetypal, satija2015spatial}, network membership estimation \citep[e.g.,][]{rubin2022statistical,huang2018pairwise,jin2024mixed}, and topic modeling \citep[e.g.,][]{ke2024using}. Existing methods in these areas focus on recovering vertices from finite-dimensional vectors. In contrast, our work extends the classical SPA to infinite-dimensional functional data. Our approach builds on prior work that improves the robustness of SPA to noise \citep{mizutani2018efficient, gillis2015semidefinite, jin2024improved}, as well as studies establishing finite-sample error bounds for SPA \citep{gillis2013fast, jin2024improved}. Leveraging these developments, we derive error bounds for the recovered basis functions under mild regularity conditions in the functional setting. While our proof strategy follows \citet{jin2024improved}, extending their analysis from finite-dimensional vectors to functions requires re-deriving several key geometric inequalities in this functional setting. Notably, the denoising step in the d-fSPA algorithm is not strictly required for consistency, but substantially improves finite-sample performance in practice.

\metahunt{} also contributes to the literature on conformal prediction in meta-analytic and multi-population settings. Foundational work on conformal prediction established distribution-free predictive inference under exchangeability  \citep{vovk2005algorithmic, shafer2008tutorial}. More recent developments have extended conformal methods to causal inference, producing intervals for counterfactual quantities \citep[e.g.,][]{lei2021conformal, jin2023sensitivity, alaa2023conformal}. In parallel, several recent papers study conformal prediction in hierarchical or multi-population settings  \citep{andrews2022transfer, dunn2023distribution, liu2024multi}, but rely on access to individual-level data from source studies. In contrast, \metahunt{} operates entirely on aggregate-level study summaries while focusing on a two-layer hierarchical data generation structure. Under mild estimation error control and regularity conditions, we establish asymptotic marginal coverage guarantees for the proposed prediction intervals, providing a practical alternative for uncertainty quantification when individual-level data are unavailable.


\section{The Setup and Assumptions}

We consider a setting motivated by aggregate-data meta-analysis, in which multiple studies provide estimated functions that capture site-specific predictive relationships. While a common example is the study of conditional average treatment effects (CATEs) across sites, our formulation applies more broadly. Rather than restricting attention to causal estimands, we treat this as a general function-valued prediction problem, where each study contributes an estimated function summarizing its local relationship between covariates and outcomes. Our objective is to leverage these study-specific estimates to predict the corresponding function in a new target population that may differ in its study-level characteristics.

Formally, suppose there exist $m$ source studies indexed by $i = 1, \ldots, m$. Each study $i$ provides an estimated function $\hat f^{(i)}(\bx): \mathcal{X} \rightarrow \R$, which may represent, for instance, a site-specific CATE function, a dose–response curve, or any other predictive object of interest. 
Associated with each study is a vector of study-level covariates $\bW_i \in \mathcal{W} \subseteq \R^p$, summarizing contextual characteristics such as region, study design, or population composition.
Let $n_i$ denote the sample size of study $i$, let $n:=\sum_{i=1}^m n_i$ denote the total number of observations across all studies, and let $n_{\min}:=\min_{1\le i\le m}n_i$ denote the smallest within-study sample size.

Our goal is to predict the corresponding function $f^{(0)}(\bx)$ for a new target population $i=0$ characterized by covariates $\bW_0$, using the collection of observed pairs $\{ (\bW_i, \hat f^{(i)}(\bx) ) \}_{i=1}^m$. Throughout, we treat the observed functions $\hat f^{(i)}(\bx)$ as noisy estimates of the underlying oracle functions $f^{(i)}(\bx)$, i.e., $\hat f^{(i)}(\bx) = f^{(i)}(\bx) + \epsilon^{(i)}(\bx)$, where $\epsilon^{(i)}(\bx)$ represents estimation error within study $i$. 
We do not require any access to individual-level data from source distribution or the target population. 
We view the studies as exchangeable units drawn from a common data-generating process, and focus on pointwise prediction for fixed covariate values $\bx \in \mathcal{X}$. 
Our goal is to learn shared structure across the $m$ studies that enables accurate out-of-sample prediction of $f^{(0)}(\bx)$ for new study-level covariate $\bW_0$.

\subsection{Low-rank assumption}
We propose the {\it low-rank cross-study heterogeneity} assumption to capture the shared structure underlying study-level oracle functions $f^{(i)}(\bx)$. Although studies may differ in complex and unobserved ways, we assume that their study-level functions share a common underlying structure captured by a small number of latent basis functions.
This assumption reduces the complexity of cross-study heterogeneity and enables us to recover a shared functional representation from the aggregate data.
\begin{assumption}[Low-rank cross-study heterogeneity]
\label{cond:low_rank}
There exists a set of basis functions $\{g_k(\bx)\}_{k=1}^K$ such that
\begin{eqnarray*}
f^{(i)}(\bx) = \sum_{k=1}^K \pi_{ik} g_{k}(\bx), \quad \forall \bx \in \mathcal{X},
\end{eqnarray*}
for all $i = 0,1,\ldots, m$ where $K < m$, and the weights satisfy $\bpi_i = (\pi_{i1}, \ldots, \pi_{iK})^\top \in \Delta_{K-1}$ with $\Delta_{K-1} = \{ \bpi \in \R^K: \sum_{k=1}^K \pi_k = 1, \min_k \pi_k \geq 0\}$ denoting the $(K-1)$-dimensional simplex. Moreover, the basis functions are non-degenerate: the differences $g_2 - g_1,\, g_3 - g_1,\, \ldots,\, g_K - g_1$ are linearly independent in $L^2(\mu)$ for a probability measure $\mu$ on $\mathcal{X}$.
\end{assumption}
Assumption~\ref{cond:low_rank} states that each study-level oracle function $f^{(i)}(\bx)$ can be expressed as a convex combination of a shared collection of latent basis functions $\{g_k(\bx)\}_{k=1}^K$. Geometrically, this means that all study-specific functions and the target function lie within the convex hull spanned by these basis functions.
The coefficients $\pi_{ik}$ represent the study-specific mixture weights that characterize each study’s position within this functional simplex, summarizing how much each latent basis function $g_k(\cdot)$ contributes to the predictive mechanism of study $i$. In this sense, the basis functions capture the fundamental ``treatment effect patterns,’’ which are combined by different studies in varying proportions.

The non-degeneracy condition ensures that no basis function can be expressed as an affine combination of the others under the $L^2(\mu)$ norm, so that the convex hull is a genuine $(K-1)$-dimensional simplex and each basis function is identifiable. The probability measure $\mu$ is the same as the one used to define the functional norm in the consistency result (Section~\ref{sec:consistency}).

The motivation behind this assumption is that, although cross-study heterogeneity may appear complex on the surface, it is often driven by a small number of persistent mechanisms. Studies usually differ along a limited set of features, such as their designs (e.g., in-lab versus online administration, or domestic versus international sites) and differences in population composition or standard of care. These features reshape the outcome--covariate relationship.  When this is the case, the study-level functions share a low dimensional structure.

The low-rank assumption is much more flexible than the standard approach of directly transporting a CATE function across sites, which requires the conditional treatment effect to be invariant given the observed individual-level covariates. Such an assumption can fail when those covariates have limited explanatory power. At the same time, the low-rank assumption is richer than restricting attention to a scalar summary such as the average treatment effect. The relaxation is structural rather than parametric, preserving the ability to flexibly model functional quantities within each study using machine learning methods.

We emphasize that the low-rank assumption is a structural modeling assumption rather than a causal identification result. Like transportability, it must be justified on subject-matter grounds. However, the assumption can be subjected to empirical scrutiny: the cross-validation criterion used to select the number of basis functions $K$ also indicates whether a low-dimensional representation adequately fits the observed studies, with a poor fit across all small values of $K$ providing evidence against the assumption. In this sense, our framework is best understood as a response-surface approach to evidence synthesis \citep{konnyu2024evidence}, as it predicts the function-valued estimate that a future study would report, rather than identifying a target-population causal effect through cross-study identification conditions.

This assumption parallels structures commonly used in archetypal analysis and topic modeling, where observations are expressed as convex combinations of a few latent archetypes or topics. 
Our assumption~\ref{cond:low_rank} can be viewed as a functional analogue of these finite-dimensional decompositions. In all such cases, a shared low-dimensional structure explains the variation observed across many high-dimensional objects, enabling efficient representation, interpretation, and generalization.

Our decomposition shares a high-level structure with classical functional representations such as Fourier expansions and functional principal component analysis (fPCA) \citep[][]{ramsay2005functional, hsing2015theoretical}. All express functions as combinations of basis functions. However, Fourier bases are fixed and universal, whereas fPCA bases, though data-driven, are orthogonal eigenfunctions of the covariance operator chosen to maximize explained variance. 
In contrast, the basis functions in \metahunt{} are data-driven and generally non-orthogonal, recovered as vertices of a convex hull. The central distinction is the simplex constraint on the mixing weights, which makes the weights $\pi_{ik}$ as membership proportions characterizing the degree to which study $i$ resembles each basis.

A closely related but distinct formulation arises in the distributionally robust optimization (DRO) literature. Recent work, such as  \cite{xiong2023distributionally} and \cite{zhang2024minimax}, defines an uncertainty set where the target function is expressed as a convex combination of the source-study functions, i.e., $\{ f^{(0)}(\bx) = \sum_{i=1}^m q_i f^{(i)}(\bx), \mathbf{q} \in \Delta_{m-1}  \}$.
While this setup is structurally similar, \metahunt{} differs in that the feasible convex hull of the target function is defined by a fixed set of latent basis functions, rather than by the observed source studies themselves. The number of bases $K$ does not grow with the number of studies $m$. Instead, these basis functions capture the common underlying structure that persists as new studies are added. This perspective emphasizes representation learning by identifying a shared generative structure, rather than expanding mixtures over an ever-growing set of empirical studies.

\subsection{Weight model}
\label{sec:weight_model}
Assumption~\ref{cond:low_rank} specifies the shared low-rank structure for study-level functions. We now describe how the corresponding mixing weights $\bpi_i$ vary across studies. These weights determine how each study combines the latent basis functions $\{g_k(\bx)\}_{k=1}^K$ to form its oracle function $f^{(i)}(\bx)$.
Since studies are conducted under different conditions, characterized by covariates such as region, time period, or experimental design, we allow the weights $\bpi_i$ to depend flexibly on the study-level covariates $\bW_i$.

\begin{assumption}[Weight model]
\label{cond:weight_model}
For all $i \in \{0, 1, \ldots, m\}$, the weight vector $\bpi_i$ is drawn independently from a conditional distribution
\begin{eqnarray*}
    \bpi_i \mid \bW_i \ \stackrel{ind.}{\sim} \ \mathcal{P}_{\bpi \mid \bW } (\cdot \mid \bW_i),
\end{eqnarray*}
where $\mathcal{P}_{\bpi \mid \bW }$ denotes an arbitrary distributional map from the covariate space $\mathcal{W}$ to the simplex $\Delta_{K-1}$. 
\end{assumption}
Assumption~\ref{cond:weight_model} treats the weight-generating process as a stochastic mechanism linking study-level covariates to the composition of the underlying basis functions. This flexible, nonparametric formulation allows the distribution of $\bpi_i$ to vary with $\bW_i$, enabling the model to capture potentially complex patterns of heterogeneity across studies.
In practice, the conditional distribution $\mathcal{P}_{\pi |W }$ can be further specified using additional parametric or semiparametric assumptions to enhance interpretability or estimation efficiency, especially when the number of studies is limited.

A flexible and principled extension is to replace the linear predictor in standard Dirichlet regression with a function that lies in a reproducing kernel Hilbert space (RKHS). This leads to a kernelized generalization of Dirichlet regression (Example~\ref{exp:dirichlet}). Such a formulation preserves the Dirichlet simplex constraint while accommodating nonlinear relationships between $\bW_i$ and the study-level weights $\bpi_i$. Kernelized Dirichlet models inherit the computational tractability of classical Dirichlet regression, while providing a semiparametric representation that can adapt to complex patterns of cross-study heterogeneity.
\begin{example}
[Kernelized Dirichlet regression model]
\label{exp:dirichlet}
Let ${\kappa}:\mathcal W \times \mathcal W \to \mathbb{R}$ be a positive–definite kernel
with associated reproducing kernel Hilbert space (RKHS) $\mathcal H_{\kappa}$ and feature
map $\psi:\mathcal W \to \mathcal H_{\kappa}$. For all
$i \in \{0,1,\ldots,m\}$, assume
\begin{eqnarray}
\bpi_i :=\left( \pi_{i1}, \pi_{i2}, \ldots, \pi_{iK} \right) \sim \mathrm{Dirichlet}(\alpha_{i1},\alpha_{i2}, \ldots, \alpha_{iK}), 
\end{eqnarray}
with $\alpha_{ik}
  = \exp\big\{\langle \boldsymbol{\beta}_k, \psi(\bW_i)\rangle_{\mathcal H_{\kappa}}\big\}$ and $\boldsymbol{\beta}_{k} \in \mathcal H_{\kappa}$ for all $k$.
\end{example}

Another approach to avoid the simplex constraint on the weights is to model their log-ratios \citep{aitc:82}. Log-ratio transformations provide an unconstrained representation of simplex-valued data and are widely used in compositional analysis. Under this formulation, the dependence of the weights on study-level covariates is modeled through a regression for the log of the relative weights, which can then be mapped back to the simplex via a softmax-type transformation. This specification preserves the interpretability of log-ratio models while allowing nonlinear covariate effects through RKHS
regularization. Example~\ref{exp:kernel_logratio} illustrates this kernelized log-ratio regression framework.
\begin{example}
[Kernelized log–ratio regression]
\label{exp:kernel_logratio}
For each $k \in \{2,\ldots,K\}$, define the 
log-ratio $\eta_{ik} := \log\!\left({\pi_{ik}}/{\pi_{i1}}\right)$.
For all $i \in \{0,1,\ldots,m\}$, assume
\begin{eqnarray}
\eta_{ik} = \langle \boldsymbol{\gamma}_k, \psi(\bW_i) \rangle_{\mathcal H_\kappa},
\qquad \boldsymbol{\gamma}_k \in \mathcal H_\kappa.
\end{eqnarray}
\end{example}

The low-rank structure (Assumption \ref{cond:low_rank}) and the weight model (Assumption \ref{cond:weight_model}) together provide a unified and interpretable representation for study-level functions. The low-rank structure captures a shared set of latent basis functions that summarize cross-study variation, while the weight model explains heterogeneity by characterizing how different studies combine these bases according to their covariates. 
This joint formulation enables generalization to new target populations and reduces the effective dimensionality of the problem.
Moreover, since both components are nonparametric, the framework retains the flexibility to capture complex, nonlinear patterns of heterogeneity while preserving interpretability through its low-rank decomposition.

A further appeal of this formulation is that the study-level information plays an active role in the prediction. Under this formulation, the prediction for the target moves toward sources conducted in similar conditions rather than toward a global average
Moreover, the framework degrades gracefully when $\bW$ carries little explanatory power. In the extreme case where $\bpi_i$ is independent of $\bW_i$, a consistently estimated weight model simply predicts the marginal mean weight, which is close to the across-study average weight. The resulting target prediction is therefore approximately the average of the projected study-specific functions, with the projection itself acting as a form of denoising. In other words, when the study-level covariates are uninformative, \metahunt{} falls back on a sensible default: an average of the denoised study-specific functions. The covariate information in $\bW_0$ thus serves only to refine this baseline toward more comparable studies.

\begin{assumption}[Exchangeability]
\label{cond:exchagneable}
The site-level covariates $\bW_0, \bW_1, \ldots, \bW_m$ are exchangeable.
\end{assumption}
Assumption~\ref{cond:exchagneable} assumes that the site-level covariates are exchangeable. Combined with Assumption~\ref{cond:low_rank}--~\ref{cond:weight_model}, this implies that $(\bW_i, \bpi_i, f^{(i)})$ are jointly exchangeable across $i= 0, 1, \ldots, m$. This exchangeability property underlies the validity of our conformal prediction procedure in Section~\ref{sec:conformal}.

\section{Estimation through Low-rank Basis Hunting}
\label{sec:estimation}

This section describes the estimation procedure. We first consider the recovery of the latent basis functions $\{g_k(\bx)\}_{k=1}^K$ that define the shared low-dimensional structure underlying the study-specific functions. We then discuss how to determine the number of basis functions $K$ in practice. Finally, given the estimated bases, we fit the weight model to estimate the study-specific mixing weights $\bpi_i$ as functions of the study-level covariates $\bW_i$. Together, these steps yield an estimated representation of study-specific functions and enable prediction for the target population.

\subsection{Basis hunting via the d-fSPA algorithm}

In this section, we extend the successive projection algorithm (SPA) from the vertex hunting setting to recover the latent basis functions $\{g_k(\bx)\}_{k=1}^K$.
We call our extension the functional successive projection (fSPA) algorithm.
The low-rank assumption implies that each study-level function $f^{(i)}(\bx)$ can be expressed as a convex combination of a small number of basis functions, placing our task in close analogy to the classical vertex hunting problem. In vertex hunting, each observed vector is modeled as a convex combination of a finite set of vertices plus noise, and the goal is to recover these vertices from noisy samples. This problem arises in a wide range of applications, including hyperspectral unmixing, archetypal analysis, network membership estimation, and topic modeling. Our setting differs from these applications in one important respect: instead of finite-dimensional vectors, our ``observations'' are functions, and the latent vertices are infinite-dimensional basis functions.

The SPA of \citet{araujo2001successive} is a simple yet powerful method for vertex hunting. It is a stepwise greedy procedure that repeatedly selects the observation with the largest norm after projecting out previously chosen directions. The algorithm is computationally efficient, requires no explicit objective function, and has desirable theoretical guarantees. 
We adopt SPA in our functional setting for two key reasons. First, SPA selects basis elements directly from the observed functions, thereby avoiding the challenge of representing an arbitrary function in a high-dimensional function space. Second, the core argument underlying SPA, namely the triangle inequality, extends naturally to function norms, allowing the algorithm’s intuition to carry over to our infinite-dimensional context.

However, SPA is a greedy procedure and can be sensitive to noise and outliers. In particular, SPA tends to select points that fall outside the underlying convex hull. This outward bias arises because noisy observations may have large norms and thus are more likely to be selected by the algorithm. To address this, it is natural to introduce a denoising step before applying SPA. Following the approach of \citet{jin2024improved}, we extend their denoising strategy to our functional setting.
The resulting denoised fSPA (d-fSPA) algorithm, presented in Algorithm~\ref{alg:spa}, provides a more robust method for basis hunting under noisy functional observations. Here, we define $B_\Delta(\hat f^{(i)}) := \{f: \Vert f - \hat f^{(i)} \Vert \leq \Delta \}$ to be the $\Delta$-neighborhood of $\hat f^{(i)}$. 

{
\begin{algorithm}[t]
\SetAlFnt{\footnotesize} 
\caption{The d-fSPA Algorithm for basis hunting \label{alg:spa}}

\KwIn{Estimated functions from $m$ studies $\{\hat f^{(i)}(\cdot)\}_{i=1}^m$; estimated number of basis $\khat$; tuning parameters $(N, \Delta)$.}

\textbf{Denoising step:}
\begin{itemize}
\zeroitem
    \item [1.] If there are fewer than $N$ functions in $B_\Delta(\hat f^{(i)})$, remove this function.
    \item [2.] Otherwise, replace $\hat f^{(i)}$ by the average of all functions in $B_\Delta(\hat f^{(i)})$.
\end{itemize}

\textbf{Initialize} $S = \emptyset$, $h_i(\cdot) =\hat f^{(i)}(\cdot)$ for $1 \leq i \leq m$

\For{$k = 1, \ldots, \khat$}{ 

1. Project $h_i$ onto the orthogonal space of $\mathrm{span}(S)$: $h_i =  h_i - P_{\mathrm{span}(S)} h_i$.

2. Find $s_k$ such that $h_{s_k}$ has the largest norm $s_k = \argmax_i \Vert h_i\Vert$. 

3. Update the set $S = S \cup \{ \hat f^{(s_k)}\}$.
}

\KwOut{Estimated basis functions $\hat g_k(\cdot) = \hat f^{(s_k)}$ for $1 \leq k \leq \khat$.}
\end{algorithm}
}

The d-fSPA algorithm includes a denoising step and a successive projection step. The denoising step reduces the influence of noise by locally averaging functions within a $\Delta$-neighborhood and discarding functions that lack sufficient nearby support. The tuning parameters $(N, \Delta)$ can be selected by cross-validation. In our numerical studies, following the recommendations of \cite{jin2024improved}, we use the heuristic choices $N= 0.5\log m$ and $\Delta = \max_{ij} \Vert \hat f^{(i)} - \hat f^{(j)} \Vert/10$.
After denoising, the d-fSPA algorithm proceeds iteratively: at each step, the algorithm projects all remaining functions onto the orthogonal complement of the span of previously selected bases, and then selects the function with the largest residual norm. Intuitively, this residual norm identifies the direction that introduces the greatest new variation, corresponding to an extreme point of the underlying simplex. The final output consists of $\khat$ selected functions that best capture the extreme structure implied by the low-rank assumption.


For the functional norm used in Algorithm~\ref{alg:spa}, we allow any norm on the relevant function space. In practice, we recommend the $L^2$ norm with respect to a probability measure $\mu$ on $\mathcal{X}$, i.e., $\Vert f - \tilde f \Vert = \{ \int ( f(\bm{x}) - \tilde f(\bm{x})  )^2\, d\mu(\bm{x}) \}^{1/2}$. A natural choice is the covariate distribution of the target population, $\mu = \mathcal{P}_{0,\bm{X}}$, which emphasizes accuracy in the regions of the covariate space most relevant to the target. If $\mathcal{P}_{0,\bm{X}}$ is unavailable, one may instead use the pooled covariate distribution from the source studies as a proxy. For the consistency guarantee established in Section~\ref{sec:consistency}, we restrict to the $L^2(\mu)$ norm.

In many applications, study-specific functions $\hat f^{(i)}$ are estimated using black-box machine learning methods and therefore lack a closed analytical form. In such cases, the $L^2$ norm can be approximated empirically: draw individual-level covariate samples, evaluate the corresponding predictions, and compute the empirical analogue of the above expectation. This approximation provides a practical and reliable way to measure functional distances when only sampled covariate values and black-box estimators are available.

\subsection{Consistency of basis recovery}
\label{sec:consistency}
We establish the consistency of the estimated basis functions $\{\hat g_k\}_{k=1}^\khat$ produced by the d-fSPA algorithm. 
The consistency result requires two conditions: one controlling the within-study estimation error and another ensuring that the observed studies are sufficiently diverse to approximate each vertex of the simplex.

The first condition controls the within-study estimation error. Recall that $\epsilon^{(i)}(\bx) := \hat f^{(i)}(\bx) - f^{(i)}(\bx)$ denotes the estimation error for study $i$. In most meta-analyses, the number of studies $m$ is moderate, so we require only that the within-study estimation error is well-controlled as the study-level sample sizes $n_i$ grow.

\begin{assumption}[Estimation error conditions]
\label{cond:est_error}
We assume the following conditions:
\begin{itemize}
\zeroitem
    \item [(a)] $\sup_{\bx \in \mathcal{X}}\E[\epsilon^{(i)}(\bx)^2] = O(n_i^{-r})$ for some $r > 0$ and all $i$.
    \item [(b)] There exists a constant $ 0 < a < r$ s.t. $m = o( \inf_i n_i^{a})$.
\end{itemize}
\end{assumption}
Assumption~\ref{cond:est_error}(a) requires the mean squared estimation error to vanish uniformly over the covariate space $\mathcal{X}$ as the within-study sample sizes grow.
Many meta-learner and machine-learning estimators satisfy this condition. For parametric estimators, we typically have $r=1$.
For non-parametric or machine-learning estimators (e.g., random forests \citep{wager2018estimation}; doubly-robust estimators \citep{kennedy2023towards}; local polynomial estimators \citep{kennedy2024minimax}; X-learner \citep{kunzel2019metalearners}), the rate $r$ depends on underlying smoothness conditions and is usually slightly below 1.
Assumption~\ref{cond:est_error}(b) ensures that the number of studies does not grow too fast relative to the within-study sample sizes.
Together, these conditions imply that the estimation error is uniformly small across all studies with probability tending to one.

The second condition controls the degree of purity among the observed studies. Since the d-fSPA algorithm recovers basis functions by selecting from the observed studies, its accuracy depends on having studies that lie sufficiently close to each vertex of the simplex.
\begin{assumption}[Approximate purity condition]
\label{cond:purity}
For each $k=1,\ldots,K$ and $\eta > 0$, define the near-pure count $M_k(\eta) := |\{1\le i\le m:\, \pi_{ik}\ge 1-\eta\}|$.
There exists a sequence $\delta_m$ s.t. $\delta_m \max_{1\leq k \leq K} \Vert g_k \Vert = o(1)$ and $\min_{1\le k\le K} M_k(\delta_m) \ge 1$.
\end{assumption}
Assumption~\ref{cond:purity} requires that for each basis function $g_k$, at least one study has its weight predominantly concentrated on $g_k$, with the approximation error $\delta_m$ vanishing as $m$ grows. When $\delta_m = 0$, there exists a pure study for every basis function, i.e., a study $i$ with $\bpi_i = e_k$ for each $k$, corresponding to the pure node condition commonly assumed in the vertex-hunting literature \citep[e.g.,][]{gillis2013fast, jin2024improved}. Assumption~\ref{cond:purity} relaxes this by requiring only approximate purity. 

Whether this condition holds generally depends on the weight model in Assumption~\ref{cond:weight_model}. For instance, a weight distribution that places non-negligible mass near each vertex of the simplex (e.g., a Dirichlet model with suitable concentration parameters) ensures that near-pure studies emerge as $m$ grows, whereas weights tightly concentrated away from the vertices may violate the condition even for large $m$. More generally, the near-pure count $M_k(\delta_m)$ characterizes how many studies lie near each vertex at the resolution $\delta_m$; this quantity plays a role in the choice of denoising parameters for the d-fSPA algorithm.

With these two conditions, we establish the following consistency guarantee.
\begin{theorem}[Consistency of d-fSPA]
\label{thm:bcons}
Suppose $\|\cdot\|$ is the $L^2(\mu)$ norm. Under Assumptions~\ref{cond:low_rank},~\ref{cond:est_error}, and~\ref{cond:purity}, with an appropriate choice of the denoising parameters $(N, \Delta)$, we have, up to permutation,
\[
\max_{1\le k\le \min(K, \khat)}
\|\hat g_k - g_k\|
=
o_P(1).
\]
\end{theorem}
When $\khat = K$, Theorem~\ref{thm:bcons} states that all $K$ estimated basis functions converge to the true basis functions, up to permutation. When $\khat > K$, only the first $K$ of the recovered bases are guaranteed to be consistent for the true basis functions.

Our proof extends the consistency analysis of \citet{jin2024improved} from the finite-dimensional setting to the functional setting. The overall strategy is the same as we first establish a finite-sample error bound, and then show that this bound vanishes under our assumptions. However, moving from vectors to functions requires re-deriving several key geometric inequalities in the functional setting.
A notable feature of our analysis is that the denoising step is not strictly required for consistency. In this sense, the estimation error is controlled by Assumption~\ref{cond:est_error} rather than by the denoising step. Nevertheless, we find that the denoising step is valuable in practice. Our numerical studies show that it  reduces finite-sample basis recovery error, particularly when $m$ is moderate or when within-study estimation error is non-negligible. 

An appropriate choice of $(N, \Delta)$ depends on the purity level $\delta_m$ and the within-study estimation error rate, and therefore implicitly on the weight model and estimation error. Intuitively, the neighborhood radius $\Delta$ should be large enough that near-vertex studies cluster together and are not mistakenly pruned, while the pruning threshold $N$ should scale with the number of near-pure studies expected at each vertex. Under the kernelized Dirichlet weight model (Example~\ref{exp:dirichlet}), for instance, the expected number of near-pure studies at each vertex satisfies $M_k(\delta_m) \asymp m\,\delta_m^{K-1}$, leading to the choices $N \asymp \log m$ and $\Delta \asymp (\log m/m)^{1/(K-1)} + n_{\min}^{(a-r)/2}$. Explicit formulas and derivations are given in the Appendix (Section~\ref{sec:denoising_proof}).


\subsection{Choice of $K$}

The proposed d-fSPA algorithm requires the estimated number of basis functions $\khat$ as an input. \metahunt{} treats $K$ as a fixed, low-dimensional parameter that does not grow with the number of studies $m$. This reflects the premise that the basis functions represent an underlying latent structure shared across studies, and this structure should remain stable as additional studies become available.

Choosing $K$ is an important modeling decision. If $K$ is too small, the resulting convex hull fails to capture key sources of cross-study heterogeneity, leading to information loss and underfitting. Conversely, selecting the value of $K$ that is too large may lead to overfitting, introduce instability when fitting the weight model, and increase prediction variance. We discuss two data-driven approaches for selecting $K$.

The first approach is to examine the following reconstruction error as a function of $K$,
\begin{eqnarray}
\label{eq:re_error}
\mathcal{E}(K) := \frac{1}{m}
\sum_{i=1}^m 
\min_{\bpi_i}\left\Vert
 \hat f^{(i)} - 
\sum_{k=1}^K \pi_{ik}\, \hat g_k
\right\Vert,
\end{eqnarray}
which measures the average distance between the observed functions and their projections onto the convex hull spanned by the $K$ selected bases. For each candidate value of $K$, we compute the corresponding reconstruction error $\mathcal{E}(K)$. Plotting this error against $K$ typically yields an “elbow” shape, where the reduction in error slows markedly beyond a certain point. This unsupervised approach is analogous to selection heuristics used in principal component analysis and archetypal analysis, and provides a convenient way to identify a reasonable range for $K$.

The second approach is to evaluate the prediction error via cross-validation for different values of $K$. Since the ultimate goal is to predict the target function $f^{(0)}(\bx)$, it is natural to select $\khat$ based on out-of-sample predictive performance. In this approach, we perform cross-validation over the entire estimation and prediction pipeline. For each candidate $K$, we estimate the basis functions using a subset of studies, fit the weight model, generate predictions for the held-out studies, and compute the prediction error. The value of $K$ that minimizes the average validation error is then selected. Although more computationally demanding, this supervised procedure is aligned with our inferential objective and appears to yield better predictive performance in practice. But, the two approaches complement each other. One can use the elbow plot to find an appropriate range for $K$, and use cross-validation to further refine the choice within that range.

\subsection{Fitting the weight model}

Given the estimated basis functions $\{ \hat g_k \}_{k=1}^\khat$ based on the d-fSPA algorithm, the next task is to estimate the study-specific mixing weights $\bpi_i$, which characterize how each study combines the shared latent basis functions.
For each study $i \in \{1, \ldots, m\}$, we obtain $\hat \bpi_i$ by projecting the observed function $\hat f^{(i)}$ onto the convex hull spanned by the estimated basis functions. Specifically, we solve the following constrained optimization problem, 
\begin{eqnarray}
\label{eq:solvepi}
        \hat \bpi_i = \argmin_{\bpi \in \Delta_{\khat-1}}  \left\Vert \hat f^{(i)} - \sum_{k=1}^\khat \pi_{ik} \hat g_{k}  \right\Vert.
\end{eqnarray}
The function norm can be chosen flexibly. In practice, we recommend the $L^2$ norm with respect to the empirical covariate distribution of the target population.
This optimization step yields the mixture weights that best reconstruct each study-level function using the recovered basis functions. 

With the estimated weights $\hat\bpi_i$ in hand, we then fit the weight model introduced in Section~\ref{sec:weight_model}. Any model that maps study-level covariates $\bW_i$ to simplex-valued weights $\bpi_i$ can be used, including parametric models (e.g. Dirichlet regression), semi-parametric models (e.g. RKHS Dirichlet regression model), or nonparametric approaches (e.g. random forests, neural networks).

A natural alternative to \metahunt{} is to pool all individual-level data across studies and fit a single regression model using both individual-level covariates $\bX$ and study-level covariates $\bW$ as joint inputs, where every individual within the same study shares the same value of $\bW$. The fitted model is then evaluated at $\bW_0$ and the target-population covariate distribution to produce a prediction for the target site. We refer to this as the pooling approach.

Compared to the pooling approach, \metahunt{} is privacy-preserving and computationally efficient. First, the pooling approach requires access to individual-level data from every source study, whereas \metahunt{} operates entirely on aggregate-level summaries $\{(\bW_i, \hat f^{(i)})\}_{i=1}^m$ and is naturally privacy-preserving. Second, the pooled model must be refit whenever a new source study is added, while \metahunt{} only updates the newly added study-level function estimate and re-runs the low-rank basis hunting step, with the costly fitting of study-specific functions performed locally at each site. Empirically, we find that \metahunt{} achieves lower RMSE than the pooling approach, possibly reflecting the implicit denoising benefit of the low-rank representation; see Section~\ref{sec:empirical} for detailed comparisons.

\section{Conformal Prediction}
\label{sec:conformal}

In Section~\ref{sec:estimation}, we explained how to obtain the recovered basis functions $\{ \hat g_k \}_{k=1}^\khat$ using the d-fSPA algorithm, and fit the weight model $\widehat{\mathcal{M}}$ that maps study-level covariates $\bW$ to mixing weights $\bpi$. With these two components, we can now predict the target function $f^{(0)}$ corresponding to a new study with covariates $\bW_0$.
Specifically, we first generate the predicted weight vector $\tilde \bpi_0$ and then construct the target function:
\begin{eqnarray} 
\tilde f^{(0)} (\cdot) = \sum_{k=1}^\khat \tilde \pi_{0k} \hat g_k (\cdot) \quad \text{where} \ \tilde \bpi_0 = \widehat{\mathcal{M}}(\bW_0).
\end{eqnarray}
The tilde notation highlights that these are predictions, in contrast to the hat notation used for estimation.

With the predicted target function, we now develop prediction intervals to quantify uncertainty through conformal prediction. 
This task presents two main challenges. First, our outcome of interest is function-valued. Some existing work \citep[e.g.,][]{lei2015conformal} develops conformal confidence bands for random functions with distribution-free guarantee.  However, these full confidence bands tend to be overly conservative and uninformative in practice. In our setting, we instead focus on pointwise inference for a fixed covariate value $\bx$, aiming to characterize uncertainty in the scalar quantity $f^{(0)}(\bx)$. 

Second, there is a measurement-error problem as the oracle study-level functions $f^{(i)}(\bx)$ are never observed directly. Instead, we only have access to their noisy estimates $\hat f^{(i)}(\bx)$.  These estimates carry estimation error that may vary across studies and propagates through all stages of \metahunt. 
To obtain theoretical guarantees for our prediction intervals, we therefore require conditions that control this estimation error.

As shown in Algorithm~\ref{alg:split_conformal}, we adapt the split conformal prediction procedure of \citet{shafer2008tutorial} to construct pointwise prediction intervals for $f^{(0)}(\bx)$ at a fixed covariate value $\bx$. 
We first split the studies into a training set and a calibration set. The training set is used to run the full estimation pipeline from Section~\ref{sec:estimation}, including basis hunting via the d-fSPA algorithm, weight estimation, and fitting the weight model. These together yield a predictor $\widehat f^{\mathrm{tr}}(\bx;\bW)$. The calibration set is then used to compute conformity scores, defined as absolute residuals between the observed study-level functions $\hat f^{(i)}(\bx)$ and their corresponding predictions. The empirical $(1-\alpha)$-quantile of these conformity scores determines the width of the prediction interval centered at the predicted value for the target study. 

\begin{algorithm}[t]
\SetAlFnt{\footnotesize} 
\caption{Split conformal prediction for $f^{(0)}(\bx)$ \label{alg:split_conformal}}

\KwIn{Estimated functions $\{\hat f^{(i)}(\cdot)\}_{i=1}^m$; target-study covariates $\bW_0$;
evaluation point $\bx$;
miscoverage level $\alpha \in (0,1)$;
estimation pipeline.}

\textbf{1. Data split:} Randomly split the index set $\{1,\ldots,m\}$ into a training set $\mathcal I_{\mathrm{tr}}$ and a calibration set $\mathcal I_{\mathrm{cal}}$, with sizes $m_{\mathrm{tr}}$ and $m_{\mathrm{cal}}$, respectively.

\textbf{2. Train predictor on training studies:} Using only $\{(\bW_i,\hat f^{(i)}(\cdot)) : i \in \mathcal{I}_{\mathrm{tr}}\}$, run the estimation pipeline of Section~\ref{sec:estimation} to obtain a prediction rule $\widehat f^{\mathrm{tr}}(\bx;\bW): \mathcal{X} \times \mathcal{W} \to \mathbb{R}$.

\textbf{3. Compute calibration residuals:} For each $i \in \mathcal{I}_{\mathrm{cal}}$, calculate the conformity score.
\begin{itemize}
\zeroitem
    \item [(i)] Compute the predicted value at $\bx$: $\tilde f^{(i)}(\bx) := \widehat f^{\mathrm{tr}}(\bx; \bW_i)$.
    \item [(ii)] Define the conformity score (absolute residual): $r_i := \bigl|\hat f^{(i)}(\bx) - \tilde f^{(i)}(\bx)\bigr|$.
\end{itemize}

\textbf{4. Estimate conformal quantile:} Let $r_{(1)} \leq \cdots \leq r_{(m_{\mathrm{cal}})}$ be the order statistics of $\{r_i : i \in \mathcal{I}_{\mathrm{cal}}\}$ and set $q_{1-\alpha}(\bx) := r_{(\lceil (1-\alpha)(m_{\mathrm{cal}}+1) \rceil)}$.

\textbf{5. Predict for the new study:} Compute the point prediction for the new study, $\tilde f^{(0)}(\bx) := \widehat f^{\mathrm{tr}}(\bx; \bW_0)$.

\KwOut{Split conformal prediction interval
\[
C_\alpha(\bx)
=
\bigl[\tilde f^{(0)}(\bx) - q_{1-\alpha}(\bx),\; \tilde f^{(0)}(\bx) + q_{1-\alpha}(\bx) \bigr].
\]}
\end{algorithm}

Split conformal prediction guarantees valid marginal coverage when the calibration data are exchangeable with the target data. In our setting, Assumption~\ref{cond:low_rank}--\ref{cond:exchagneable} ensure exchangeability of $(\bW_i, f^{(i)}(\bx))$ across studies. Additionally, we require mild control over the estimation error in the study-level functions $\hat f^{(i)}(\bx)$ as stated in Assumption~\ref{cond:est_error}.

Finally, the following theorem establishes that our conformal interval achieves asymptotic marginal coverage for the target quantity $f^{(0)}(\bx)$. 
\begin{theorem}[Asymptotic marginal coverage guarantee]
\label{thm:ci_cvg}
Under Assumption~\ref{cond:low_rank}--~\ref{cond:est_error}, we have
\begin{eqnarray*}
    \lim \Pr( f^{(0)}(\bx) \in C_\alpha(\bx) ) \geq 1- \alpha.
\end{eqnarray*}
\end{theorem}
This guarantee holds despite the existence of estimation error in the study-level functions and the fact that prediction proceeds through a multi-stage procedure involving basis hunting and weight modeling. 
The key reason is that conformal prediction 
relies mainly on the exchangeability of calibration residuals.
Because Assumption~\ref{cond:est_error} ensures that estimation errors are uniformly small, even though the prediction rule is learned through a complex pipeline, the conformal adjustment corrects for its uncertainty in a distribution-free manner, ensuring valid coverage in large samples.

It is worth contrasting our conformal procedure with the pooling approach discussed at the end of Section~\ref{sec:estimation} from the perspective of exchangeability. The pooling approach treats $\bX$ and $\bW$ symmetrically as joint covariates in a single regression model. Depending on how conformal prediction is applied, this can implicitly require the exchangeability of individual-level triples $(\bX_{ij}, \bW_i, Y_{ij})$ across both studies and individuals within studies, which is a stronger condition than the study-level exchangeability of $\{(\bW_i, f^{(i)})\}$ used by \metahunt. Empirically, the pooling approach tends to produce shorter prediction intervals, consistent with its access to higher-resolution data; see Section~\ref{sec:empirical} for detailed comparisons.

The split conformal procedure in Algorithm~\ref{alg:split_conformal} produces intervals of a single width $q_{1-\alpha}(\bx)$ that does not depend on $\bW_0$, which can be conservative when calibration residuals vary across studies. Several extensions of conformal prediction can potentially improve efficiency by allowing  the width to adapt to $\bW_0$. Weighted conformal prediction \citep{tibshirani2019conformal} reweights calibration residuals by kernel similarity to $\bW_0$. Conformalized quantile regression \citep{romano2019conformalized} conformalizes intervals from a conditional quantile regression on $\bW$. Both approaches retain a marginal coverage guarantee, though weighted conformal additionally requires the weights to coincide with the true likelihood ratio.

\section{Extension to Functionals of the Target Function}

\metahunt{} naturally extends from predicting the target function $f^{(0)}$ to predicting functionals of $f^{(0)}$. 
Let $\Phi : \mathcal{F} \rightarrow \mathbb{R}$ be a functional mapping elements of the function space $\mathcal{F}$ to real numbers.
Our object of interest is then $\theta^{(0)} = \Phi\big(f^{(0)}\big)$.
This formulation accommodates many quantities of substantive interest, including averages, integrals, and other policy-relevant summaries of the predicted function.

As a leading example, consider the average treatment effect (ATE) for the target population.  If $f^{(0)}(\bx)$ represents a conditional average treatment effect (CATE), the target ATE is given by
\[
\tau^{(0)}
=
\Phi\big(f^{(0)}\big)
=
\mathbb{E}_{\mathcal{P}_{0,X}}
\!\left[
f^{(0)}(X)
\right],
\]
where $\mathcal{P}_{0,X}$ denotes the covariate distribution of the target population.

Given the predicted target function $\tilde f^{(0)}$, a natural point estimator of $\theta^{(0)}$ is the plug-in estimator $\tilde \theta^{(0)} = \Phi\big(\tilde f^{(0)}\big)$.
In the ATE example, this corresponds to
\[
\tilde\tau^{(0)}
=
\mathbb{E}_{\mathcal{P}_{0,X}}
\!\left[
\tilde f^{(0)}(X)
\right].
\]

Uncertainty quantification for $\theta^{(0)}$ can be obtained by modifying the conformity score in Algorithm~\ref{alg:split_conformal}.
For each calibration study $i$, define
\[
r_i
=
\left|
\Phi\big(\hat f^{(i)}\big)
-
\Phi\big(\tilde f^{(i)}\big)
\right|.
\]
That is, we measure the discrepancy between the observed functional summary and its predicted counterpart. Prediction intervals for $\theta^{(0)}$ are then constructed using the empirical quantile of $\{r_i\}$, yielding a split-conformal interval centered at $\tilde \theta^{(0)}$.

\begin{corollary}[Asymptotic marginal coverage for functionals]
\label{cor:functional_cvg}
Suppose Assumptions~\ref{cond:low_rank}--~\ref{cond:est_error} hold. In addition, suppose $\Phi$ is Lipschitz continuous with respect to $\|\cdot\|_{L^2(\mathcal{P}_{0,\bm{X}})}$, i.e., there exists a constant $L_\Phi > 0$ such that
\begin{eqnarray*}
|\Phi(f) - \Phi(g)| \leq L_\Phi \|f - g\|_{L^2(\mathcal{P}_{0,\bm{X}})}, \quad \text{for all } f, g \in \mathcal{F}.
\end{eqnarray*}
Then
\begin{eqnarray*}
    \lim \Pr\!\left( \theta^{(0)} \in C_\alpha^\Phi \right) \geq 1 - \alpha,
\end{eqnarray*}
where $C_\alpha^\Phi := [\tilde\theta^{(0)} - q_{1-\alpha}^\Phi,\; \tilde\theta^{(0)} + q_{1-\alpha}^\Phi]$ is the split conformal prediction interval with $q_{1-\alpha}^\Phi$ denoting the $\lceil(1-\alpha)(m_{\mathrm{cal}}+1)\rceil$-th smallest value among the conformity scores $\{r_i: i \in \mathcal{I}_{\mathrm{cal}}\}$.
\end{corollary}
The Lipschitz condition is satisfied by many functionals of practical interest. For example, when $\Phi(f) = \mathbb{E}_{\mathcal{P}_{0,\bm{X}}}[f(\bm{X})]$ corresponds to the ATE, the Cauchy--Schwarz inequality gives $L_\Phi = 1$.

\section{Simulation Studies}
\label{sec:simulation}
We conduct simulation studies to evaluate the finite-sample performance of \metahunt.

\subsection{Setup}

We generate data from a two-layer hierarchical model using $m = 100$ source studies.  Although the baseline simulation uses $m=100$, we also explicitly vary the number of source studies from smaller to larger values to assess sensitivity to the number of available studies. 
In the first layer, we generate the study-level covariates $\bW_i \in \R^3$, the mixing weights $\bpi_i$,  and the corresponding oracle functions $f^{(i)}$. For each study $i$, the study-level covariates and mixing weights are generated as
\[
\begin{aligned}
\bW_i &\sim \mathcal{N}(\bm 0, \bm I_3), \\[-0.3em]
\boldsymbol{\alpha}_i &= 20\exp\!\big( (1, \bW_i)^\top \boldsymbol{\beta} \big), \\[-0.3em]
\bpi_i &\sim \mathrm{Dirichlet}(\boldsymbol{\alpha}_i),
\end{aligned}
\qquad\qquad
\boldsymbol{\beta} = 
\begin{pmatrix}
1 & -1 & -1 & 1 \\
2 & 1 & -1 & -1 \\
0 & 4 & 0 & 0 \\
1 & 0 & -3 & 0
\end{pmatrix}.
\]
Each study-level oracle function is then constructed as a convex combination of four basis functions
\begin{eqnarray*}
 f^{(i)}(x)
&=& \sum_{k=1}^4 \pi_{ik}\, g_k(x), \\
g_1(x)&=&-2x+3,\quad
g_2(x)=x^2/4,\quad
g_3(x)= 10 \sin(x/3),\quad
g_4(x)= -2|x+4|.   
\end{eqnarray*}

In the second layer, we generate individual-level data 
\((X_{ij}, Y_{ij})\) for each study, using a common sample size 
\(n_i = 200\). Specifically, $X_{ij} \sim \N(0, 5^2)$, and $Y_{ij} \mid X_{ij} \sim \N( f^{(i)}(X_{ij}), 5^2 )$.
For each study \(i\), we then estimate the study-level function 
\(\hat f^{(i)}\) using a random forest regression fitted to its individual-level data.  These estimated functions \(\{\hat f^{(i)}\}_{i=1}^m\) serve as the inputs to our low-rank basis hunting procedure based on the d-fSPA algorithm.

For implementation, we approximate the functional norm in the d-fSPA algorithm using the empirical $L^2$ norm with respect to the covariate distribution of the data-generating process, i.e., $\N(0,5^2)$. In practice, this is computed by evaluating each function on an independently sampled grid of covariate values and taking the empirical average of squared differences. For the denoising step in the d-fSPA algorithm, we set the tuning parameters to $N=0.5 \log m$ and $\Delta = \max_{ij} \Vert \hat f^{(i)} - \hat f^{(j)} \Vert/10$, which scale the neighborhood radius relative to the observed variability among the estimated functions.
After the basis hunting step, we estimate the study-specific mixing weights via constrained projection and fit the weight model using a Dirichlet regression linking the estimated weights to the study-level covariates $\bW_i$. This pipeline yields the fitted components required for prediction and subsequent conformal inference. All reported mean squared error and coverage results are based on 200 independent Monte Carlo simulation runs.

\subsection{Findings}

As shown in Figure~\ref{fig:chooseK}, we assess the choice of the number of basis functions $K$ using two proposed criteria: reconstruction error and cross-validation prediction error. The left panel reports the reconstruction error $\mathcal{E}(K)$. For both fSPA and d-fSPA algorithms, the error decreases sharply for small $K$ and then levels off, producing an elbow. The elbow appears near $K=4$ for both methods. The middle panel shows cross-validation prediction errors, which decline substantially up to the true value $K=4$ and then rise gradually for larger values. This pattern is consistent with the usual bias–variance tradeoff, as the model with small $K$ underfits and fails to capture cross-study heterogeneity, while an overly large $K$ introduces instability in the weight estimation step without improving predictive accuracy. Across both criteria, the d-fSPA algorithm generally achieves lower errors than the fSPA algorithm, reflecting the benefits of the denoising step in basis hunting.
The right panel shows the distribution of CV-selected $\khat$ across 200 simulation runs, where $K=4$ is selected most frequently. Taken together, these results support selecting $\khat=4$ and using the d-fSPA algorithm in subsequent analyses.

\begin{figure}[!t]
    \centering
    \includegraphics[width=\linewidth]{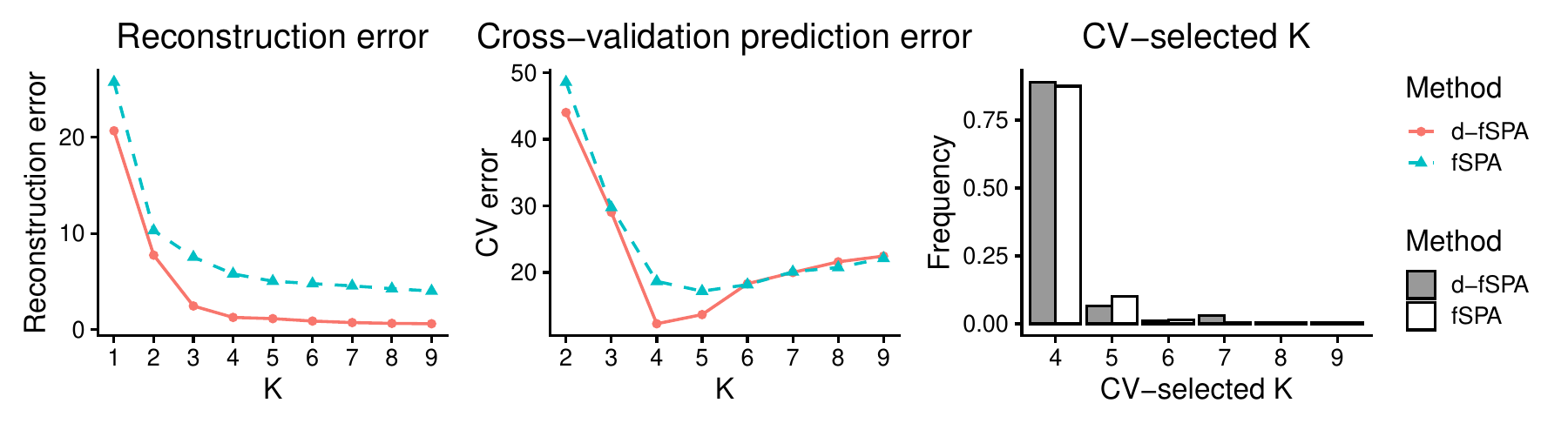}
    \caption{
    (Left panel) Reconstruction error decreases with $K$ and exhibits an elbow.
    (Middle panel) Cross-validation prediction error across values of $K$, minimized at $K=4$ for d-fSPA.
    (Right panel) Empirical distribution of cross-validation–selected $K$ over 200 Monte Carlo runs, with the true value $(K=4)$ selected most frequently.}
    \label{fig:chooseK}
\end{figure}

Figure~\ref{fig:predm} reports the mean squared error (MSE) of the d-fSPA algorithm as a function of the number of studies $m$, for several choices of the basis dimension $K$. 
For implementation of the denoising step, we set the tuning parameter $N=0.5 \log m$, and choose the neighborhood radius parameter $\Delta = \max_{ij} \Vert \hat f^{(i)} - \hat f^{(j)} \Vert/\{0.2 \times (m-100)+10\}$. This choice coincides with the neighborhood radius used above when $m=100$.
For all choices of $K$, the MSE decreases rapidly as $m$ increases from relatively small sample size, followed by a more gradual improvement thereafter. When $K$ is too small ($K=2$), the MSE remains substantially higher, indicating underfitting due to an insufficiently expressive set of basis functions. The correctly specified value ($K=4$) consistently achieves the lowest MSE, demonstrating that the d-fSPA algorithm is able to recover and exploit the true low-rank structure. Further increasing $K$ leads to higher prediction error, suggesting that overly large basis sets introduce redundancy and geometric mis-specification that degrade predictive performance. Overall, these results indicate that the d-fSPA algorithm achieves optimal predictive accuracy at the true value of $K$. 

\begin{figure}[!t]
    \centering
    \includegraphics[width=0.5\linewidth]{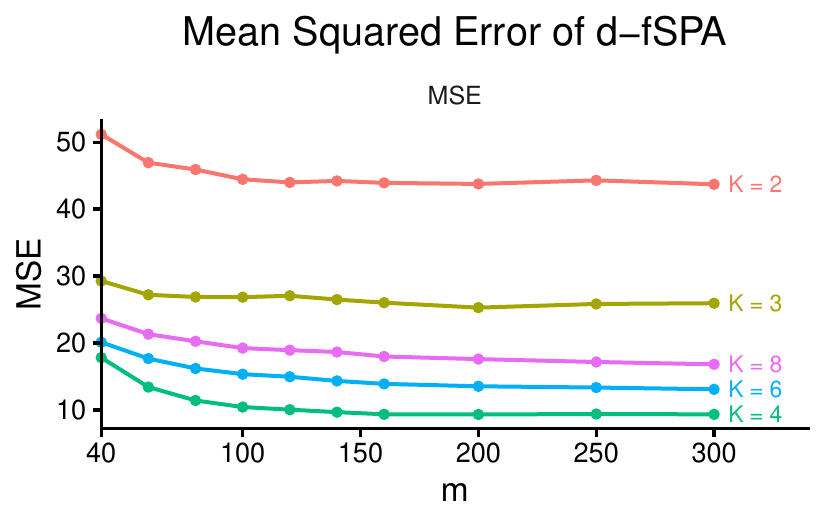}
    \caption{
    Mean squared error (MSE) of d-fSPA as a function of the number of studies $m$ over 200 simulation runs. MSE is shown for different number of the bases. For all choices of $K$, MSE decreases as $m$ increases, with the correctly specified value $(K=4)$ achieving the lowest error. }
    \label{fig:predm}
\end{figure}

We next evaluate the finite-sample performance of our conformal prediction procedure. We implement a 70-30 train-calibration split and estimate coverage over 200 simulation runs. In each run, we generate predictions for 100 target study-level covariates $\bW_0$ and 1000 values of $X$. The overall empirical coverage across the 200 runs is 98.43\%, exceeding the 95\% target. This mild overcoverage is driven by estimation error in the study-level functions $\hat f^{(i)}$, which introduces additional uncertainty beyond what is accounted for in the conformal procedure. 
Figure~\ref{fig:cvg} reports the average coverage and average 95\% prediction interval length versus different values of $x$. Theorem~\ref{thm:ci_cvg} guarantees asymptotic marginal coverage for any fixed $x$ under exchangeability and mild estimation error control. In finite samples, coverage remains close to or above the nominal level for most values of $x$, while undercoverage may occur at extreme $x$ values, where estimation error is larger, and data are relatively sparse. The right panel shows that confidence interval lengths are shortest in regions with high data density and expand toward the tails, reflecting increased uncertainty and demonstrating that the conformal intervals adapt to local variability.

\begin{figure}[!t]
    \centering
    \includegraphics[width=0.7\linewidth]{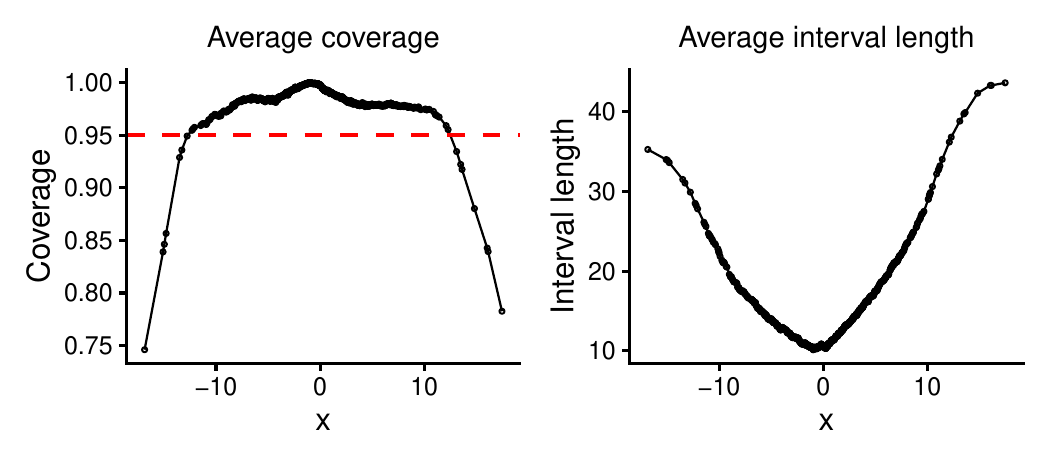}
    \caption{Average marginal coverage (left) and 95\% prediction interval length (right) versus $\bx$ over 200 simulation runs. The red dashed line in the left panel denotes the nominal 95\% coverage. Coverage remains above the nominal level for most values of $x$, while it decreases and prediction intervals widen at extreme values of $x$, reflecting increased estimation uncertainty in the tails.}
    \label{fig:cvg}
\end{figure}

Figure~\ref{fig:ci} visualizes prediction intervals for six representative target study covariates $\bW_0$. The shaded regions represent the prediction intervals. The predicted values (red dots) generally track the true values (blue curves) closely, indicating that the combination of basis hunting and weight modeling captures the main functional patterns well. The prediction intervals adapt to local uncertainty and widen in regions with greater variability. Across all panels, most true values are covered by the intervals. These results suggest that the proposed conformal procedure provides reliable uncertainty quantification despite the multi-stage estimation pipeline and the estimation error in the observed study-level functions $\hat f^{(i)}$.

\begin{figure}[!t]
    \centering
    \includegraphics[width=.9\linewidth]{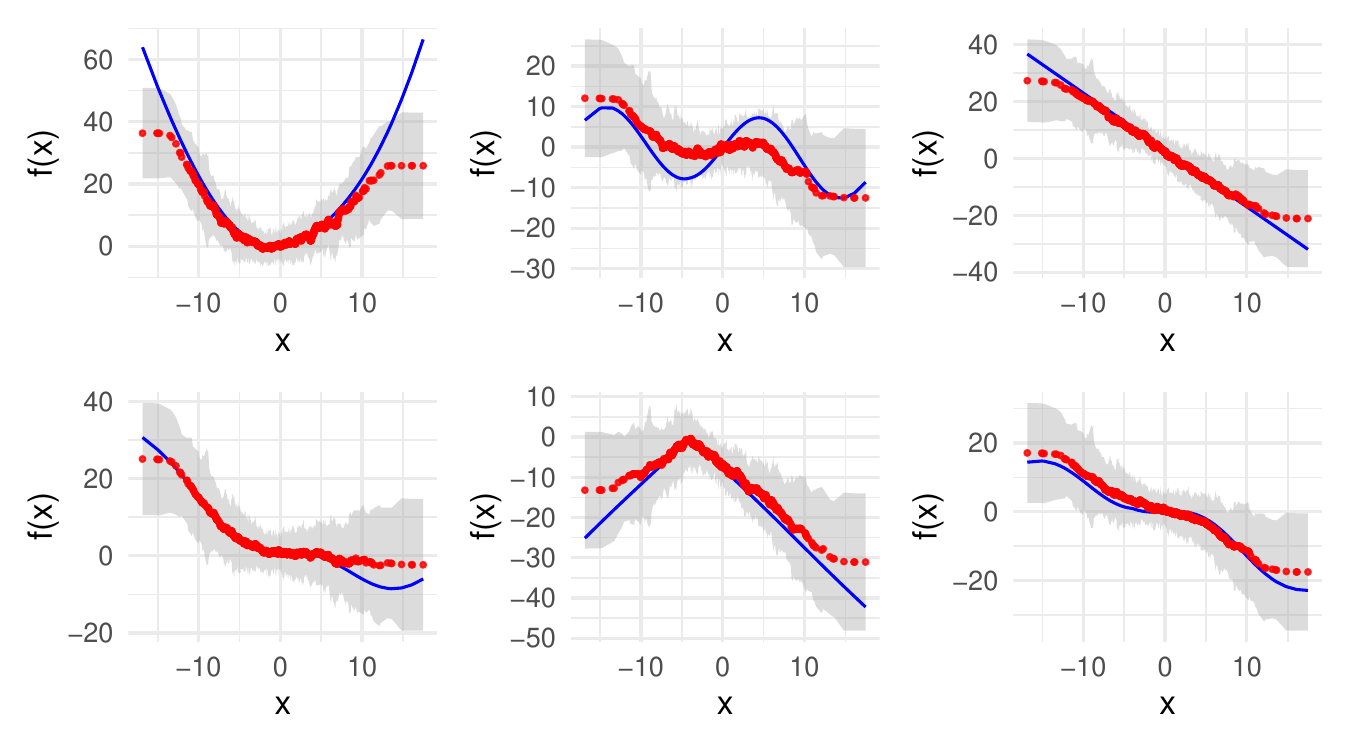}
    \caption{
    Prediction intervals, predicted values, and expected values versus $\bx$ for six different new study-level covariates $\bW_0$. The predicted functions closely track the expected values across the domain, while the prediction intervals adapt to local variability and widen in regions with greater uncertainty. }
    \label{fig:ci}
\end{figure}

\section{Empirical Application}
\label{sec:empirical}

We illustrate \metahunt{} using the Many Labs 1 project \citep{klein2014investigating}. 

\subsection{Setup}
The Many Labs 1 project is a large-scale multisite replication study in psychology. The project evaluates the replicability of 13 classic and contemporary hypotheses under a standardized protocol, with data collected across $m=36$ independent sites. Among these, 27 studies were conducted in laboratory settings and 9 online; 25 sites are located in the United States and 11 in other countries. The combined dataset contains a total of 6,344 participants. 
The hypotheses span a range of well-known effects in social and cognitive psychology, in which participants estimate a numerical quantity, such as the population of Chicago or the height of Mount Everest, after exposure to a high or low numerical anchor. Outcomes vary across hypotheses and live on different natural scales. A complete list of hypotheses is provided in Supplementary Materials Section~\ref{sec:ml1_hypotheses}.
We preprocess the raw data following the scripts provided by the original authors.

The analysis incorporates covariates at two levels. At the site level, we use a vector $\bW_i \in \mathbb{R}^6$ summarizing study context. This includes indicators for online versus laboratory administration and US versus international site, as well as site-level summary statistics such as average age, gender ratio, average political ideology, and a scalar measure of covariate shift relative to the target site constructed following \citet{jin2025beyond}. At the individual level, we use $\bX \in \mathbb{R}^5$ consisting of demographic and attitudinal variables: gender, age, race, political ideology, and American identity. 

We consider prediction and causal inference tasks. In both cases, the ultimate target is a scalar functional of the target-site function $f^{(0)}$.
\begin{itemize}
\item[1.] \textbf{Prediction task (average outcome).} 
For each site $i$ and hypothesis, we estimate a site-specific outcome regression function $\hat f^{(i)}$ using a random forest. The target estimand is the target-site mean outcome.
\item[2.] \textbf{Causal inference task (ATE).} 
For each site $i$ and hypothesis with a binary treatment, we estimate a site-specific CATE function $\hat f^{(i)}$ using causal forests. The target estimand is the target-site average treatment effect (ATE).
\end{itemize}

We implement a leave-one-site-out evaluation scheme. For each target site $i$, we fit the model using the remaining sites and then predict the corresponding target-site quantity. Because the project design provides the data from each target site, we can directly evaluate generalization performance in terms of prediction error and empirical coverage. As an empirical benchmark for each held-out target site, we use the raw target-site estimate (the sample mean for the prediction task and the raw difference-in-means ATE for the causal task), enabling transparent out-of-sample comparison.

Due to the small number of sites, all conformal prediction intervals in the empirical application are constructed using cross-conformal prediction \citep{vovk2018cross}.
In addition, for both \metahunt{} and the pooling approach, we report a weighted conformal variant based on \citet{tibshirani2019conformal}, in which each calibration site $i$ is assigned a Gaussian kernel weight $\exp(-\|\bW_i - \bW_0\|^2 / (2h^2))$ based on its similarity to the target site, and the conformal quantile is replaced by the corresponding weighted quantile. The bandwidth is set to $h = 3 \cdot h_\text{med}$, where $h_\text{med}$ is the median of pairwise distances in $\bW$. This variant upweights sites closer to the target and produces tighter intervals than the original procedure.

In the empirical application, we benchmark \metahunt{} against several existing generalization strategies. For the prediction task, we compare Methods 1--3 below. For the causal inference task, we compare all six methods.
\begin{itemize}
    \item [1.] {\bf \metahunt}: We perform d-fSPA algorithm with the number of basis functions $K$ selected by cross-validation, and fit the weight model via Dirichlet regression using site-level covariates. The resulting estimator combines study-specific functions through the learned low-rank representation and weight model.
    \item [2.] {\bf Minimax-regret}: 
    The minimax-regret aggregation method of \citet{zhang2024minimax}, which constructs a weighted average of site-specific models by minimizing the worst-case regret over an uncertainty set. Prediction intervals are obtained via conformal inference.
    \item [3.] {\bf Pooling}:
    For this approach discussed at the end of Section~\ref{sec:estimation}, we use random forests for the prediction task and causal forests for the causal inference task. Prediction intervals are obtained via conformal inference.
    \item [4.] {\bf Partial-pooling}: A one-stage individual-participant-data meta-analysis based on a hierarchical linear model fit to the pooled source data \citep{burke2017meta}, regressing the outcome on individual-level covariates $\bX$, treatment, and their interactions, with a random intercept and a random treatment effect per site. We use the \texttt{lme4} package to fit the model \citep{bates2015lme4}. Prediction intervals are constructed using the corresponding analytical variance estimates.
    \item [5.] {\bf Random-effects}:
    This meta-analytic baseline computes the raw difference-in-means estimator at each site. These site-level estimates are then aggregated via linear random-effects meta-regression, incorporating site-level covariates $\bW$ \citep{higgins2009re}. We use the \texttt{metafor} package to fit the random-effects model \citep{viechtbauer2010conducting}. Prediction intervals are constructed using the corresponding analytical variance estimates.
    \item [6.] {\bf Doubly-robust}:
    This two-step procedure is based on doubly robust generalization estimators \citep{dahabreh2019generalizing}. At each site, we first compute a doubly robust estimator that adjusts for covariate shift in $\bX$. In our implementation, both the outcome regression and the site-membership propensity score are estimated using regression forests.
    The resulting estimates are then combined via random-effects meta-regression with site-level covariates $\bW$. Prediction intervals are obtained using analytical variance formulas.
    \item [7.] {\bf Entropy-balancing}:
    This two-step procedure applies entropy balancing \citep{hainmueller2012entropy} at each site to adjust for covariate shift in $\bX$. The adjusted site-level estimates are then aggregated using random-effects meta-regression on $\bW$. Prediction intervals are constructed based on analytical variance estimates.
\end{itemize}
These methods differ in how they use individual and site-level information. The pooling and partial-pooling approaches require access to individual-level data and therefore do not operate under a privacy-preserving evidence synthesis framework. The minimax-regret method aggregates site-specific models without incorporating site-level information $\bW$. In contrast, the meta-analytic approaches (Random-effects, Doubly-robust, and Entropy-balancing) explicitly adjust for $\bW$ through meta-regression. For the causal inference task, we further select a subset of site-level covariates $\bW$ via cross-validation {using \metahunt{}, and the same subset is then applied to all methods}. Finally, although the analytical variance formulas do not account for the uncertainty from the within-study estimation, we use them to assess the severity of this problem, as researchers use them frequently in practice. 

\subsection{Findings}

\subsubsection{Prediction task}

\begin{figure}[!t]
    \centering
    \includegraphics[width=.9\linewidth]{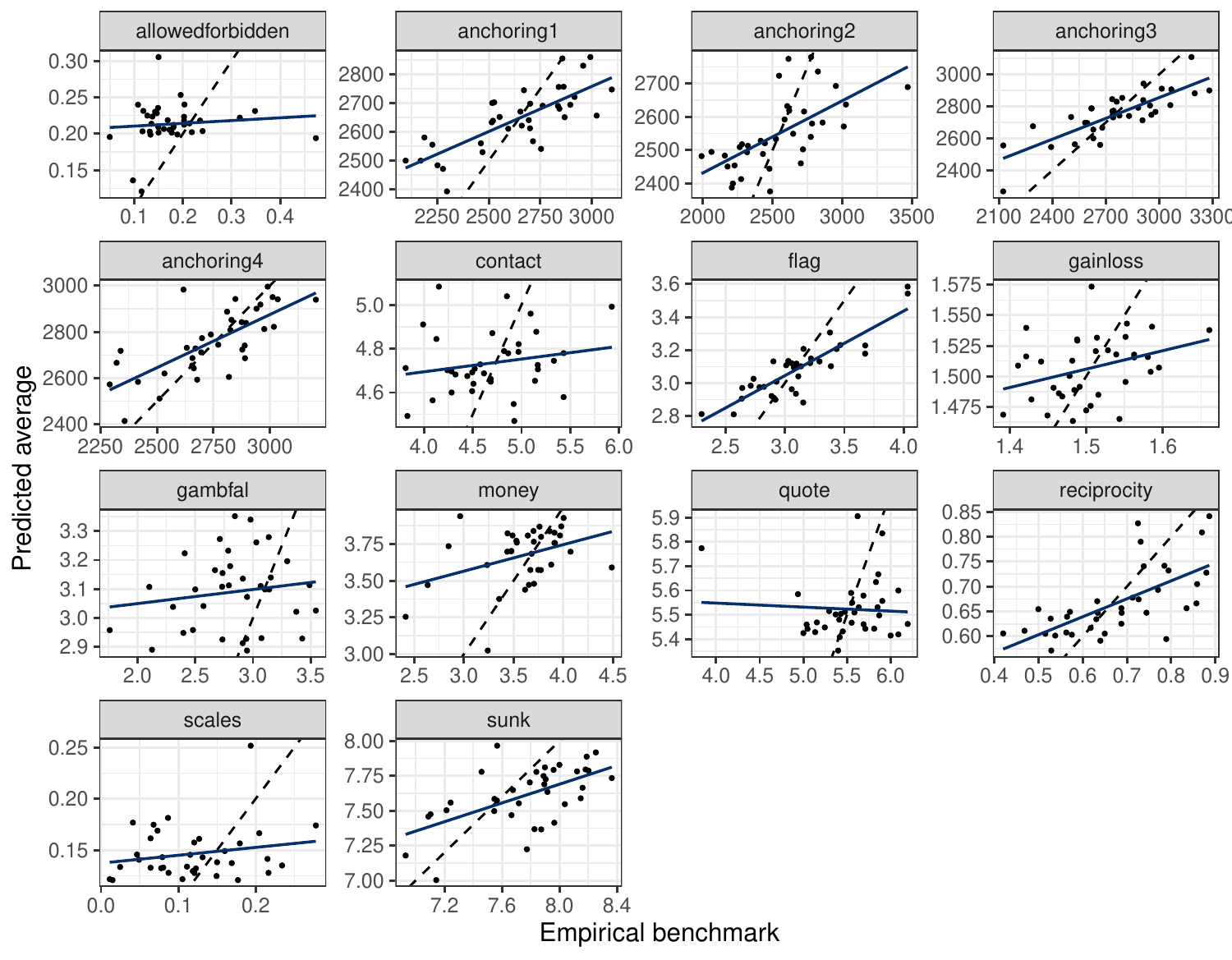}
    \caption{
    Scatterplots of predicted target-site mean outcomes versus the corresponding empirical benchmarks (raw target-site sample means), displayed separately for each hypothesis using \metahunt. The dashed 45-degree line corresponds to perfect prediction, while the solid line shows the fitted linear regression.}
\label{fig:pred_scatter}
\end{figure}

Figure~\ref{fig:pred_scatter} summarizes the performance of \metahunt{} for the prediction task by plotting predicted target-site mean outcomes against the corresponding empirical benchmarks, separately for each hypothesis.
Performance varies across hypotheses. For some (e.g., various \textit{anchoring} and \textit{flag} hypotheses), predictions align closely with the 45-degree line, indicating strong accuracy, while others exhibit greater dispersion, occasionally driven by outlying sites (e.g., \textit{gambfal} and \textit{quote}).
This heterogeneity is expected, as the explanatory power of site-level covariates $\bW$ differs across settings. Nevertheless, across many hypotheses, the relationship between predictions and benchmarks remains broadly positive, indicating that the method captures systematic cross-site variation in outcomes.

Table~\ref{tab:pred_comparison} compares the overall performance of the three methods for the prediction task. Prediction accuracy is measured by the standardized RMSE, defined as the root mean squared prediction error normalized by the empirical standard deviation of the target-site averages, allowing aggregation across hypotheses with different scales. A value of one corresponds to using the average of empirical benchmarks across all sites (including the target site) as the predictor.

\metahunt{} achieves the lowest standardized RMSE (0.904), indicating the best overall predictive accuracy. The pooling approach attains the smallest bias (0.203), with \metahunt{} a close second (0.273). Together with \metahunt{}'s lower RMSE, this indicates that the low-rank representation trades a small amount of bias for a larger reduction in variance, yielding better overall accuracy despite not requiring individual-level data. The minimax-regret approach is dominated on both criteria, consistent with its focus on worst-case regret rather than targeted prediction.

Conformal prediction is applied to all methods and achieves coverage above the nominal 95\% level.
The final column reports the average interval length normalized by the maximum interval width across methods. With the unweighted conformal procedure, all three methods produce similarly conservative intervals; the weighted conformal variant tightens the interval length of \metahunt{} from $0.969$ to $0.764$ (a 21\% reduction) while maintaining coverage above the 95\% target, and further reduces the interval length of the pooling approach from $0.806$ to $0.635$.

\begin{table}[t]
\centering
\begin{tabular}{lcccc}
\toprule
\textbf{Method} & \textbf{Std. Bias} & \textbf{Std. RMSE} & \textbf{Coverage} & \textbf{Interval Length (relative)} \\
\midrule
\metahunt       & \multirow{2}{*}{0.273} & \multirow{2}{*}{0.904} & 0.988 & 0.969 \\
\quad (weighted PI)       &             &             & 0.976 & 0.764 \\
Minimax-regret  & 0.611 & 1.125          & 0.982 & 0.964 \\
Pooling           & \multirow{2}{*}{0.203} & \multirow{2}{*}{0.933}          & 0.974 & 0.806 \\
\quad (weighted PI)       &             &             & 0.944 & 0.635 \\
\bottomrule
\end{tabular}
\caption{Overall performance comparison for prediction task. ``Std.\ Bias'' is the bias standardized by the empirical standard deviation of the target-site averages. ``Std.\ RMSE'' is the root mean squared prediction error, standardized by the same way. Column ``Coverage'' gives the empirical coverage rate of the nominal 95\% prediction intervals. ``Interval Length (relative)'' denotes the average prediction interval length normalized by the maximum interval width across methods. For \metahunt{} and the pooling approach, ``(weighted PI)'' reports coverage and interval length under the weighted conformal procedure.}
\label{tab:pred_comparison}
\end{table}



\subsubsection{Causal inference task}

Figure~\ref{fig:ate_scatter} summarizes the performance of \metahunt{} for the causal inference task. It shows the target-site ATE estimates against the corresponding empirical benchmarks, separately for each hypothesis.
As in the prediction task, performance varies across hypotheses, reflecting differences in how informative the site-level covariates $\bW$ are for explaining cross-site heterogeneity. Overall, the alignment between predictions and benchmarks is weaker than in the prediction task, with smaller correlations and greater dispersion.  This indicates the increased difficulty of ATE generalization relative to outcome prediction.

\begin{figure}[!t]
    \centering
    \includegraphics[width=.9\linewidth]{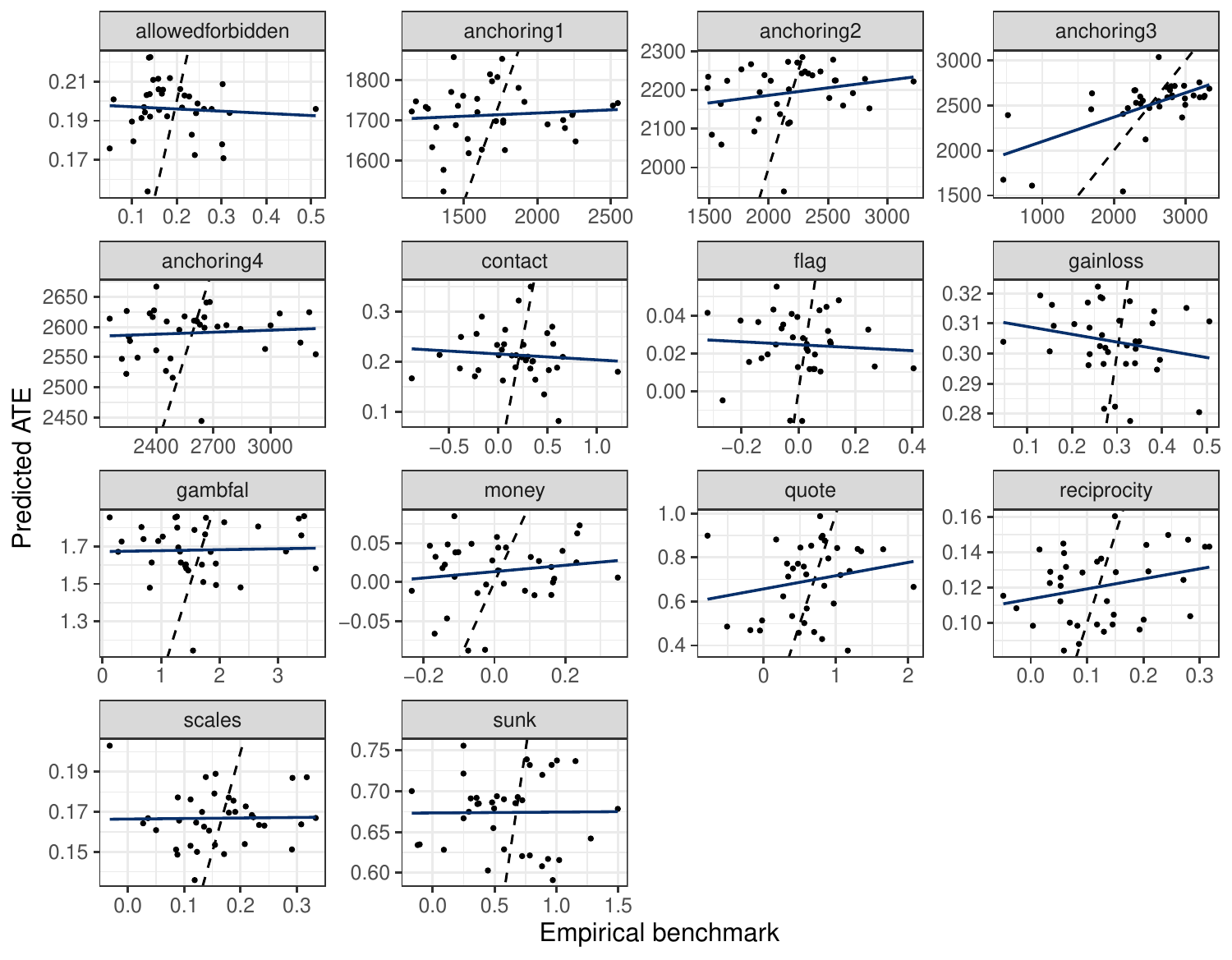}
    \caption{Scatterplots of predicted target-site ATE versus the corresponding empirical benchmarks (raw target-site ATE), displayed separately for each hypothesis using \metahunt. The dashed 45-degree line corresponds to perfect prediction, while the solid line shows the fitted linear regression.}
\label{fig:ate_scatter}
\end{figure}

Table~\ref{tab:ate_comparison} compares the overall performance of all methods. \metahunt{} achieves both the lowest standardized RMSE (0.981) and the lowest standardized bias (0.090). The two individual-data benchmarks, pooling and partial-pooling, match \metahunt{} closely on bias (0.094 and 0.099, respectively) but have higher RMSE (0.992 and 1.013), despite their access to individual-level data. All other methods give RMSEs above one; minimax-regret in particular yields a substantially higher bias (0.371), as it does not incorporate study-level information.

All methods based on conformal prediction achieve coverage above the nominal 95\% level. In contrast, methods relying on analytical variance estimates (Partial-pooling, Random-effects, Doubly-robust, and Entropy-balancing) exhibit undercoverage. This finding is expected and consistent with that of \citet{jin2024mixed}, suggesting that these models are unable to fully account for the heterogeneity across studies.
Among methods with valid coverage, \metahunt{} with the weighted conformal variant attains the shortest interval length (0.660) while maintaining coverage above 95\%. The weighted variant of pooling method produces even shorter intervals (0.598) but dips slightly below the nominal level (0.921).

\begin{table}[t]
\centering
\begin{tabular}{lcccc}
\toprule
\textbf{Method} & \textbf{Std. Bias} & \textbf{Std. RMSE} & \textbf{Coverage} & \textbf{Interval Length (relative)} \\
\midrule
\metahunt       & \multirow{2}{*}{0.090} & \multirow{2}{*}{0.981} & 0.976 & 0.801 \\
\quad (weighted PI)       &            &            & 0.956 & 0.660 \\
Minimax-regret  & 0.371 & 1.076          & 0.972 & 0.781 \\
Pooling           & \multirow{2}{*}{0.094} & \multirow{2}{*}{0.992}          & 0.958 & 0.735 \\
\quad (weighted PI)       &             &             & 0.921 & 0.598 \\
Partial-pooling   & 0.099 & 1.013         & 0.724 & 0.369 \\
Random-effects        & 0.130 & 1.058          & 0.667 & 0.305 \\
Doubly-robust         & 0.186 & 1.134          & 0.734 & 0.438 \\
Entropy-balancing       & 0.244 & 1.657          & 0.837 & 0.934 \\
\bottomrule
\end{tabular}
\caption{Overall performance comparison for the ATE task. ``Std.\ Bias'' is the bias standardized by the empirical standard deviation of the target-site averages. ``Std.\ RMSE'' is the root mean squared prediction error, standardized by the same way. ``Coverage'' is the empirical coverage of the nominal 95\% prediction intervals, and ``Interval Length (relative)'' is the average interval length normalized by the maximum across methods. For \metahunt{} and the pooling approach, ``(weighted PI)'' reports coverage and length under the weighted conformal procedure.}
\label{tab:ate_comparison}
\end{table}

Lastly, Figure~\ref{fig:ate_cmp} compares the hypothesis-specific performance for the causal inference task. Performance varies across hypotheses. \metahunt{} and the pooling method typically achieve standardized MSE around or below one, while other approaches rarely do so, indicating weaker ability to capture cross-site heterogeneity.
Entropy-balancing is notably unstable, with large variability in MSE and generally wider intervals, likely due to difficulties in constructing stable weights under substantial covariate differences across sites.
Methods based on conformal prediction maintain near-nominal coverage across hypotheses, whereas those using analytical variance estimates exhibit systematic undercoverage.

\begin{figure}[!t]
    \centering
    \includegraphics[width=\linewidth]{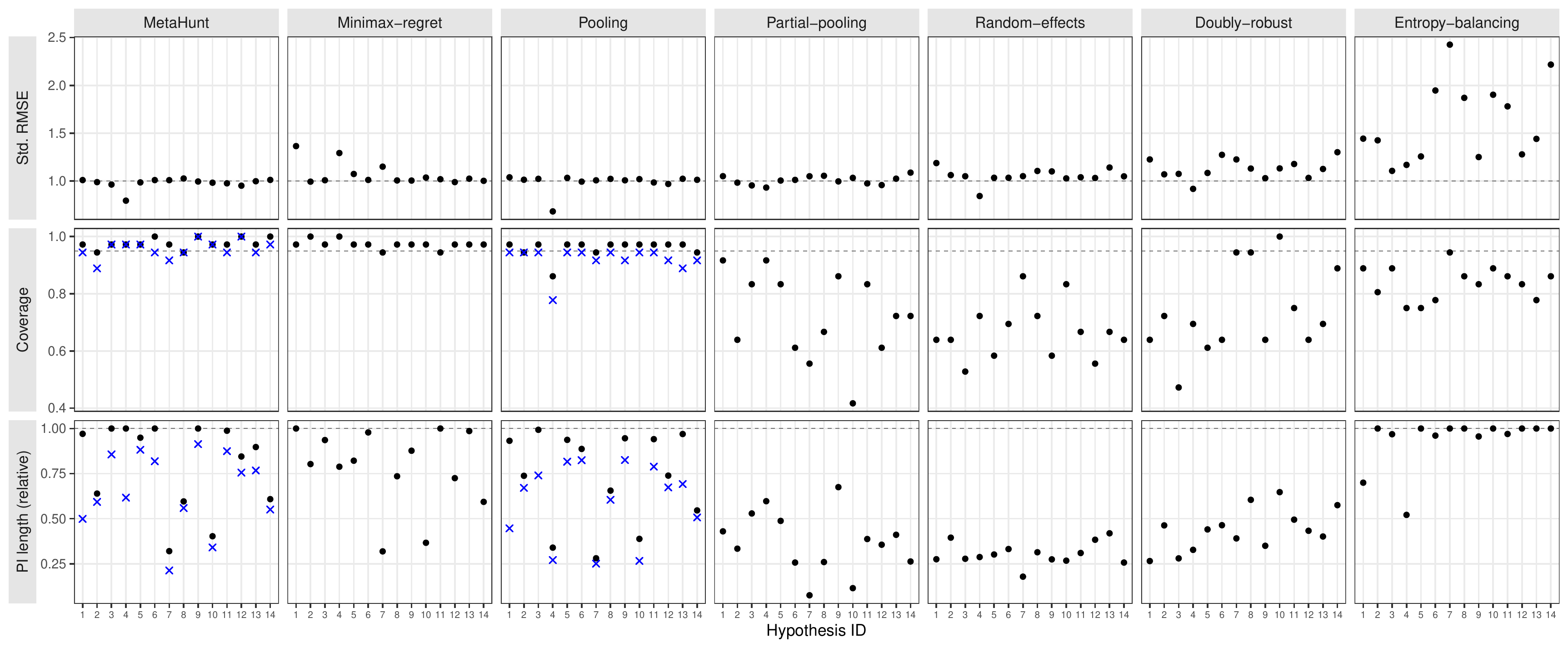}
    \caption{Hypothesis-specific performance for the causal inference task under leave-one-site-out evaluation. The top, middle, and bottom rows report standardized root mean squared error (RMSE), empirical coverage of nominal 95\% prediction intervals, and relative prediction interval length, respectively, for each method across hypotheses. Each black point corresponds to one hypothesis. In the coverage and interval-length rows, blue crosses overlay the values obtained under the weighted conformal procedure for the two methods (\metahunt{} and pooling) where this variant is available. The dashed horizontal line at 1 in the top row denotes the empirical variance of the target-site benchmark. Interval lengths are normalized by the maximum interval width across methods within each hypothesis to facilitate comparison of efficiency. }
    \label{fig:ate_cmp}
\end{figure}



\section{Concluding Remarks}

We have proposed \metahunt, a framework for privacy-preserving meta-analysis of function-valued quantities across heterogeneous studies. \metahunt{} combines a low-rank basis representation, recovered by the proposed d-fSPA algorithm, with a flexible weight model linking mixing weights to study-level covariates. We established formal consistency of the recovered basis functions, extending vertex-hunting analysis to the infinite-dimensional functional setting for the first time. For uncertainty quantification, we developed conformal prediction intervals with asymptotic marginal coverage for both pointwise prediction and scalar functionals such as the average treatment effect. Because \metahunt{} operates entirely on aggregate-level summaries, it is naturally privacy-preserving and accommodates modern machine learning estimators for the study-level functions. In our empirical application on the Many~Labs~1 data, \metahunt{} achieves lower prediction RMSE than all competing approaches, including the pooled approach that has access to individual-level data, while maintaining valid conformal coverage above the nominal level.

We note several limitations of \metahunt. First, the d-fSPA algorithm recovers basis functions by selecting from the observed studies, and its accuracy depends on the availability of near-pure studies whose mixing weights are concentrated near a single vertex of the simplex. When the number of studies $m$ is small, such near-pure studies may not be present, degrading basis recovery. A similar issue arises when the target's site-level covariates $\bW_0$ lie outside the range of the observed $\bW_i$, since the weight model may then predict target weights near a vertex that no source study covered. A small $m$ also makes reliable selection of $K$ more difficult, as the reconstruction error and cross-validation criteria become less stable. In addition, the weight model is at risk of overfitting when $m$ is small, depending on its complexity.

Second, the predictive performance relies on the study-level covariates $\bW$ being sufficiently informative about the mixing weights $\bpi$. When $\bW$ has limited explanatory power, the weight model may fail to capture meaningful patterns of cross-study heterogeneity. This was reflected in the empirical application, where the causal inference task exhibited weaker performance than the outcome prediction task, partly due to the difficulty of explaining cross-site variation in treatment effects using the available site-level covariates. Third, although we provide theoretically justified choices of the denoising tuning parameters $(N, \Delta)$, the resulting basis recovery can still be sensitive to these choices in finite samples, especially when the number of studies $m$ is small.

We highlight several directions for future research. First, in some applications such as multi-center clinical trials, outcome data may be available from the target site in addition to the study-level covariates $\bW_0$. Incorporating such target-site observations could yield efficiency gains, especially when the CATE function is estimated with substantial noise. For example, we can directly constrain the mixing weights or by refining the predicted function through a semi-supervised update. How to optimally combine target-site data with the proposed aggregate-level framework is an open problem. Second, the exchangeability assumption underlying our conformal prediction procedure may not hold when studies are collected sequentially, exhibit systematic temporal trends, or are selected in a non-representative manner. Extending the framework to non-exchangeable settings, for example through online conformal methods or time-varying weight models, would broaden its applicability. 

\bigskip
\bibliography{main}  

\clearpage
\appendix
\setcounter{equation}{0}
\setcounter{theorem}{0}
\setcounter{lemma}{0}
\renewcommand {\theequation} {A\arabic{equation}}
\renewcommand {\thetheorem} {A\arabic{theorem}}
\renewcommand {\thelemma} {A\arabic{lemma}}

\section{Many Labs 1 Hypotheses}
\label{sec:ml1_hypotheses}

Table~\ref{tab:ml1_hypotheses} provides the full names, original references, and outcome types for the experimental findings of the Many Labs 1 project \citep{klein2014investigating} analyzed in Section~\ref{sec:empirical}. The codes match those used in the original data release.

\begin{table}[h]
\centering
\small
\begin{tabular}{lll}
\toprule
\textbf{Code} & \textbf{Effect (original reference)} & \textbf{Outcome type} \\
\midrule
\texttt{anchoring1}      & Anchoring -- distance SF--NYC \citep{jacowitz1995measures}     & numerical estimate \\
\texttt{anchoring2}      & Anchoring -- Chicago pop.\ \citep{jacowitz1995measures}         & numerical estimate \\
\texttt{anchoring3}      & Anchoring -- Everest height \citep{jacowitz1995measures}        & numerical estimate \\
\texttt{anchoring4}      & Anchoring -- Babies' age \citep{jacowitz1995measures}           & numerical estimate \\
\texttt{sunk}            & Sunk Costs \citep{oppenheimer2009instructional}                 & Likert rating \\
\texttt{gainloss}        & Gain vs.\ Loss Framing \citep{tversky1981framing}               & binary choice \\
\texttt{flag}            & Flag Priming \citep{carter2011single}                           & Likert rating \\
\texttt{contact}         & Imagined Contact \citep{husnu2010elaboration}                   & Likert rating \\
\texttt{money}           & Currency Priming \citep{caruso2013mere}                         & Likert rating \\
\texttt{gambfal}         & Retrospective Gambler's Fallacy \citep{oppenheimer2009retrospective} & categorical choice \\
\texttt{quote}           & Quote Attribution \citep{lorge1936prestige}                     & Likert rating \\
\texttt{reciprocity}     & Norm of Reciprocity \citep{hyman1950current}                    & Likert rating \\
\texttt{allowedforbidden}& Allowed vs.\ Forbidden \citep{rugg1941experiments}              & Likert rating \\
\texttt{scales}          & Low-vs-High Category Scales \citep{schwarz1985response}         & ordinal category \\
\texttt{iat}             & Sex diff.\ in implicit math attitudes \citep{nosek2002math}     & continuous (D-score) \\
\bottomrule
\end{tabular}
\caption{Experimental findings from the Many Labs 1 project \citep{klein2014investigating} analyzed in this paper, identified by the short codes used in the original data release. The \emph{Anchoring} effect of \citet{jacowitz1995measures} is included in four versions, differing only in the quantity that subjects are asked to estimate.}
\label{tab:ml1_hypotheses}
\end{table}

\paragraph{Descriptions of the hypotheses.} For completeness, we briefly describe each effect below.
\begin{itemize}
\item \texttt{anchoring1} -- Anchoring, distance SF--NYC \citep{jacowitz1995measures}:
Participants estimate the distance between San Francisco and New York City after first being exposed to either a high or low numerical anchor. Even when the anchor is arbitrary, it biases the subsequent numerical judgment, and the hypothesis predicts that estimates shift in the direction of the anchor.
\item \texttt{anchoring2} -- Anchoring, population of Chicago \citep{jacowitz1995measures}:
Participants estimate the population of Chicago after receiving either a high or low anchor value. As in the other anchoring tasks, estimates are pulled toward the provided reference, and the hypothesis predicts significantly higher estimates in the high-anchor condition than in the low-anchor condition.
\item \texttt{anchoring3} -- Anchoring, height of Mt.\ Everest \citep{jacowitz1995measures}:
Participants estimate the height of Mt.\ Everest after exposure to a high or low anchor. Despite the implausibility of the anchor in some conditions, it influences the final judgment, illustrating that even relatively knowledgeable domains are vulnerable to contextual numerical cues.
\item \texttt{anchoring4} -- Anchoring, babies born per day \citep{jacowitz1995measures}:
Participants estimate the number of babies born daily in the United States after exposure to a high or low anchor. As in the other anchoring tasks, estimates systematically differ by anchor condition.
\item \texttt{sunk} -- Sunk Cost Effect \citep{oppenheimer2009instructional}:
Participants decide whether to continue an activity after having paid for it versus receiving it for free. The hypothesis predicts greater persistence when prior costs have been incurred, illustrating an irrational escalation of commitment based on unrecoverable past investment.
\item \texttt{gainloss} -- Gain vs.\ Loss Framing \citep{tversky1981framing}:
Participants choose between programs framed in terms of lives saved (gains) or lives lost (losses). Although the outcomes are logically equivalent, the framing affects risk preference: the hypothesis predicts greater risk aversion under gain framing and greater risk seeking under loss framing.
\item \texttt{flag} -- Flag Priming \citep{carter2011single}:
Participants are subtly exposed to images containing the American flag before reporting political attitudes. The hypothesis predicts that exposure to the national flag shifts respondents toward more conservative policy endorsement, testing whether symbolic cues can nonconsciously influence political preferences.
\item \texttt{contact} -- Imagined Contact \citep{husnu2010elaboration}:
Participants are asked to imagine a positive interaction with a member of an outgroup. Imagined intergroup contact is hypothesized to increase willingness to engage with that group in the future, testing whether mental simulation alone can reduce prejudice and improve intergroup attitudes.
\item \texttt{money} -- Currency Priming \citep{caruso2013mere}:
Participants are exposed to images of currency before reporting attitudes toward economic systems. The hypothesis predicts that exposure to money increases endorsement of free-market arrangements and tolerance of social inequality.
\item \texttt{gambfal} -- Retrospective Gambler's Fallacy \citep{oppenheimer2009retrospective}:
Participants observe an unlikely streak of events (e.g., repeated sixes in dice rolls) and then estimate prior occurrences. The hypothesis predicts that witnessing an unlikely streak increases estimates of how many prior rolls occurred, reflecting a tendency to reconstruct the past to make rare events appear more probable.
\item \texttt{quote} -- Quote Attribution \citep{lorge1936prestige}:
Participants evaluate the same quotation attributed either to a liked or a disliked figure. The hypothesis predicts greater agreement when the quote is attributed to a positively regarded source, demonstrating how source credibility shapes attitude evaluation.
\item \texttt{reciprocity} -- Norm of Reciprocity \citep{hyman1950current}:
Participants evaluate whether rights should be granted to an outgroup, either before or after considering similar rights for their ingroup. The hypothesis predicts increased support for granting rights to the outgroup when reciprocal rights for the ingroup are made salient first, reflecting a norm of fairness and consistency.
\item \texttt{allowedforbidden} -- Allowed vs.\ Forbidden \citep{rugg1941experiments}:
Participants are asked whether a controversial behavior should be ``allowed'' versus ``forbidden.'' Although the two wordings are logically equivalent, the framing systematically shifts responses: people are less likely to endorse forbidding an action than to reject allowing it.
\item \texttt{scales} -- Low- vs.\ High-Category Scales \citep{schwarz1985response}:
Participants report a behavior (e.g., television watching) using response scales with different numeric ranges. The range of the response options shifts how respondents interpret their own behavior, producing higher self-reported frequencies when the scale implies that higher values are normative.
\item \texttt{iat} -- Sex Differences in Implicit Math Attitudes \citep{nosek2002math}:
Participants complete an Implicit Association Test linking mathematical and gender concepts. The hypothesis predicts that participants exhibit implicit associations consistent with the cultural stereotype that math is male, with the strength of the effect differing by participant gender.
\end{itemize}

\section{Proofs}
\subsection{Proof of Theorem~\ref{thm:bcons}}

To prove this theorem, we begin by introducing the notation and describing our proof strategy.  We then establish a series of lemmas before proving the theorem.

\subsubsection{Notation and proof strategy}

Let $(\mathcal H,\langle\cdot,\cdot\rangle)$ be the Hilbert space $L^2(\mu)$ equipped with the inner product induced by the functional norm, and write $\|h\|_{\mathcal H} := \langle h,h\rangle^{1/2}$. Throughout this section, $\|\cdot\|_{\mathcal H}$ denotes the same $L^2(\mu)$ norm as $\|\cdot\|$ used in the main text. We note that $\mathcal H$ here refers to $L^2(\mu)$ and should not be confused with the RKHS $\mathcal{H}_\kappa$ appearing in Examples~\ref{exp:dirichlet} and~\ref{exp:kernel_logratio}.
Let $G:\mathbb R^K\to\mathcal H$ be the linear operator defined by $G e_k=g_k$, where $\{e_k\}_{k=1}^K$ denotes the canonical basis of $\mathbb R^K$. Define the Gram matrix $\Gamma\in\mathbb R^{K\times K}$ by $\Gamma_{ij}:=\langle g_i,g_j\rangle$.

We center the basis functions by
\[
\tilde g_k := g_k-\bar g,\quad \text{for } k=1,\dots,K \quad \text{where} \quad \bar g := \frac1K\sum_{k=1}^K g_k.
\]
Then, for any $\pi,\tilde\pi\in\Delta_{K-1}$,
\[
\sum_{k=1}^K(\pi_k-\tilde\pi_k)g_k =
\sum_{k=1}^K(\pi_k-\tilde\pi_k)\tilde g_k,
\]
since $\sum_{k=1}^K(\pi_k-\tilde\pi_k)=0$. Hence, only the centered family $\{\tilde g_k\}_{k=1}^K$ matters on the simplex-difference subspace.
Let $\tilde G:\mathbb R^K\to\mathcal H$ be the linear operator defined by $\tilde G e_k=\tilde g_k$, and define the corresponding Gram matrix $\tilde\Gamma\in\mathbb R^{K\times K}$ by $\tilde\Gamma_{ij}:=\langle \tilde g_i,\tilde g_j\rangle $.
Denote $\sigma_\ast := \sqrt{\lambda_{K-1}(\tilde\Gamma)}$, where $\lambda_k(\Gamma)$ denotes the $k$-th largest eigenvalue of a symmetric matrix $\Gamma$.

Denote the convex hull generated by the basis functions by
\begin{eqnarray*}
\mathcal S := \Big\{\sum_{k=1}^K \pi_k g_k:\ \pi\in\Delta_{K-1}\Big\},
\end{eqnarray*}
and define the following geometric quantities:
\begin{eqnarray}
\beta(\hat f,G)
&:=& \max\Bigg\{
\max_{1\le i\le m}\mathrm{Dist}_{\mathcal H}(\hat f^{(i)}, \mathcal S),\
\max_{1\le k\le K}\ \min_{1 \leq i \leq m}\ \|\hat f^{(i)}-g_k\|_{\mathcal H}
\Bigg\}, \label{eq:beta_def}\\
\gamma (G) &:=& \min_{g_0 \in \mathcal S} \max_{1 \leq k \leq K}  \Vert g_k - g_0 \Vert_{\mathcal H} , \label{eq:gamma_def}\\
d_{\max}(G)
& := & \max_{1 \leq k \leq K}  \Vert g_k  \Vert_{\mathcal H} , \label{eq:dmax_def}
\end{eqnarray}
where $\mathrm{Dist}_{\mathcal H}(f,\mathcal S):=\inf_{s\in\mathcal S}\|f-s\|_{\mathcal H}$.
The quantity $\beta(\hat f,G)$ measures how well the observed functions approximate the simplex geometry: the first term captures the maximum distance of any observed function from $\mathcal S$, and the second captures the worst-case distance from any vertex to the nearest observed function. The quantity $\gamma(G)$ is the minimax radius of the simplex, and $d_{\max}(G)$ is the maximum norm of the basis functions.

We first note that the non-degeneracy condition in Assumption~\ref{cond:low_rank} is equivalent to $\sigma_\ast > 0$.
Indeed, since the differences $g_k - g_1 = \tilde g_k - \tilde g_1$ lie in the span of $\{\tilde g_1,\ldots,\tilde g_K\}$, and conversely $\tilde g_1 = -K^{-1}\sum_{k=2}^K(g_k - g_1)$ shows each centered function lies in the span of the differences, we have $\mathrm{span}\{g_2-g_1,\ldots,g_K-g_1\} = \mathrm{span}\{\tilde g_1,\ldots,\tilde g_K\}$.
Therefore, the differences are linearly independent in $L^2(\mu)$ if and only if $\mathrm{rank}(\tilde\Gamma)=K-1$, which is equivalent to $\sigma_\ast = \sqrt{\lambda_{K-1}(\tilde\Gamma)}>0$.

The proof proceeds in three parts. We first establish a non-asymptotic error bound for the fSPA algorithm (without denoising) in the Hilbert space $\mathcal H$ (Theorem~\ref{thm:spa_pure}). Extending the analysis of \citet{jin2024improved} from $\mathbb R^d$ to $\mathcal H$, this result shows that the basis recovery error is bounded by a constant multiple of $\beta(\hat f,G)$, provided $\beta(\hat f,G)$ is sufficiently small relative to $\sigma_\ast$. We then show that under Assumptions~\ref{cond:est_error} and~\ref{cond:purity}, the quantity $\beta(\hat f,G)$ vanishes in probability (Lemma~\ref{lem:beta_control}), bridging the finite-sample bound to the consistency result. Finally, in the proof of Theorem~\ref{thm:bcons}, we incorporate the denoising step by showing that the denoised functions satisfy $\beta(\tilde f,G) = o_P(1)$ with appropriate denoising parameters $(N,\Delta)$, from which consistency follows by the same argument.

\begin{theorem}[fSPA error bounds]
\label{thm:spa_pure}
Suppose $\sigma_\ast >0$ and $\beta(\hat f,G)$ satisfies
\begin{equation}
\label{eq:thmS1.1}
450\, d_{\max}(G)\, \max\Big\{1,\ \frac{d_{\max}(G)}{\sigma_\ast}\Big\}\, \beta(\hat f,G)
\;\le\; \sigma_\ast^2.
\end{equation}
Apply the fSPA (Algorithm~\ref{alg:spa} without the denoising step) to $\hat f^{(1)},\dots,\hat f^{(m)}\in\mathcal H$ and let $\hat g_1,\dots,\hat g_K$ denote the $K$ selected functions. Then, up to a permutation of these $K$ functions,
\begin{equation}
\label{eq:thmS1.2}
\max_{1\le k\le K}\ \|\hat g_k - g_k\|_{\mathcal H}
\;\le\;
\Bigg(1 + \frac{30\,\gamma(G)}{\sigma_\ast}\max\Big\{1,\ \frac{d_{\max}(G)}{\sigma_\ast}\Big\}\Bigg)
\,\beta(\hat f,G).
\end{equation}
\end{theorem}

\subsubsection{Lemmas}
To prove Theorem~\ref{thm:spa_pure}, we first establish a series of lemmas.
\begin{lemma}
\label{lem:S1}
For any $\pi,\tilde\pi\in\Delta_{K-1}$,
\begin{equation}
\label{eq:lemmaS1.1}
\sqrt{\lambda_{K-1}(\tilde\Gamma)}\,\|\pi-\tilde\pi\|_2
\;\le\;
\|G \pi -G \tilde\pi \|_{\mathcal H}
\;\le\;
\gamma(G)\,\|\pi-\tilde\pi\|_1.
\end{equation}
Moreover, fix an integer $1\le s\le K-2$. If $\pi$ and $\tilde\pi$ share at least $s$ common entries,
then
\begin{equation}
\label{eq:lemmaS1.2}
\|G\pi-G\tilde\pi\|_{\mathcal H}
\;\ge\;
\sqrt{\lambda_{K-1-s}(\tilde\Gamma)}\,\|\pi-\tilde\pi\|_2.
\end{equation}
\end{lemma}

\begin{lemma}
\label{lem:S2}
Fix an integer $K \ge 2$ and functions $ g_1,\ldots, g_K\in\mathcal H$.
Define
\[
a \;:=\; \min_{i\neq j}\| g_i- g_j\|_{\mathcal H},
\qquad
b \;:=\; \max_{i\neq j}\Big|\|g_i\|_{\mathcal H}-\| g_j\|_{\mathcal H}\Big|.
\]
For any $\pi=(\pi_1,\ldots,\pi_K)\in\Delta_{K-1}$, let $L \;:=\; \sum_{k=1}^K \pi_k \| g_k\|_{\mathcal H}$.
Then
\begin{equation}
\label{eq:func-lemmaB2}
\biggl\|\sum_{k=1}^K \pi_k g_k \biggr\|_{\mathcal H}
\;\le\;
L \;-\; \frac{a^2-b^2}{4L}\sum_{k=1}^K \pi_k(1-\pi_k).
\end{equation}
\end{lemma}

\begin{lemma}
\label{lem:S3}
Fix $1\le s\le K-2$. Let $ P:\mathcal H\to\mathcal H$ be any orthogonal projection operator
whose range is $s$-dimensional. Define the $K\times K$ matrix
\[
\Gamma_{P} \;:=\; G^\ast (I-P)\, G,
\qquad\text{i.e.,}\qquad
(\Gamma_{ P})_{ij}=\langle (I- P)g_i,\ (I- P)g_j\rangle .
\]
Then
\begin{equation}
\label{eq:func-lemmaB3}
\lambda_{K-1-s}(\Gamma_{ P}) \;\ge\; \lambda_{K-1}(\Gamma).
\end{equation}
\end{lemma}

\begin{lemma}
\label{lem:S4}
Fix $0 \le s\le K-2$. Suppose there exist indices $\{k_1,\dots,k_s\}\subset\{1,\dots,K\}$ such that $\|g_{k_j}\|_{\mathcal H}\le \delta$ for $j=1,\dots,s$.
If $\lambda_{K-1-s}(\Gamma)\ \ge\ 2(K-2)\delta^2$, then
\begin{equation}
\label{eq:func-lemmaB4}
\max_{1\le k\le K}\|g_k\|_{\mathcal H}
\;\ge\;
\frac{\sqrt{K-s-1}}{\sqrt{2(K-s)}}\,
\sqrt{\lambda_{K-1-s}(\Gamma)}
\;\ge\;
\frac12\,\sqrt{\lambda_{K-1-s}(\Gamma)}.
\end{equation}
\end{lemma}


Given $\varepsilon\in(0,1)$ and $1\le k\le K$.  
Define the $\varepsilon$-simplicial neighborhood of $g_k$ inside $\mathcal S$ as
\[
\mathcal V_k(\varepsilon)
:=
\Big\{
f\in\mathcal S:\ 
f=\sum_{\ell=1}^K \pi_\ell g_\ell,\ \pi\in\Delta_{K-1},\ \pi_k \ge 1-\varepsilon
\Big\}.
\]
Define
\[
\mathcal K^\ast
:=
\Big\{k\in\{1,\dots,K\}:\ \|g_k\|_{\mathcal H}=d_{\max}(G)\Big\}.
\]
Fix $h_0>0$ and $\varepsilon_0\in(0,1/2)$. Define
\[
\mathcal K(h_0)
:=
\Big\{k\in\{1,\dots,K\}:\ \|g_k\|_{\mathcal H}\ge d_{\max}(G)-h_0\Big\},
\]
and
\begin{eqnarray}
\mathcal V(\varepsilon_0,h_0)
:=
\bigcup_{k\in\mathcal K(h_0)} \mathcal V_k(\varepsilon_0)
\;\subset\;\mathcal S.
\label{eq:V-eps-h}    
\end{eqnarray}

\begin{lemma}
\label{lem:S5}
Suppose there exists $\sigma_\ast>0$ such that
\begin{equation}
\label{eq:func-B22}
d_{\max}(G)\ \ge\ \sigma_\ast/2,
\qquad
\min_{1\le k\neq \ell\le K}\|g_k-g_\ell\|_{\mathcal H}
\ \ge\ \sqrt{2}\,\sigma_\ast .
\end{equation}
For any $t>0$ such that $\max\{1,\ d_{\max}(G)/\sigma_\ast\}\,t \;<\ 3\sigma_\ast$,
set
\begin{equation}
\label{eq:lemS5-const}
h_0 \;=\; \sigma_\ast/3,
\qquad
\frac12 > \varepsilon_0 \;\ge\; 6\sigma_\ast^{-1}
\max\{1,\ d_{\max}(G)/\sigma_\ast\}\,t .
\end{equation}
Then
\begin{equation}
\label{eq:func-B24}
\|f\|_{\mathcal H}
\;\le\;
d_{\max}(G)-t,
\qquad
\text{for all } f\in \mathcal S\setminus\mathcal V(\varepsilon_0,h_0).
\end{equation}
\end{lemma}

The following lemma controls $\beta(\hat f,G)$ and is used to bridge Theorem~\ref{thm:spa_pure} to the consistency result (Theorem~\ref{thm:bcons}).
\begin{lemma}
\label{lem:beta_control}
Under Assumptions~\ref{cond:low_rank},~\ref{cond:est_error}, and~\ref{cond:purity}, we have
\[
\beta(\hat f,G) \;\le\; \max_{1\le i\le m}\|\epsilon^{(i)}\|_{\mathcal H} \;+\; 2\,\delta_m\, d_{\max}(G)
\;=\; o_P(1).
\]
\end{lemma}

\subsubsection{Proof of Theorem~\ref{thm:spa_pure}}

We follow the proof strategy of Theorem B.1 in \cite{jin2024improved}.
The proof consists of two steps. In Step 1, we study the first iteration of SPA and show that $\hat g_1$ falls in the neighborhood of a true basis function. In Step 2, we recursively study the remaining iterations.

\paragraph{Step 1: Analysis of the first iteration of SPA.}

We begin by analyzing the first iteration of SPA. For notational convenience, write
\[
\gamma := \gamma(G),\qquad
d_{\max} := d_{\max}(G),\qquad
\beta := \beta(\hat f,G).
\]
Recall that, by the definition of $\beta(\hat f,G)$ given in Equation~\eqref{eq:beta_def},
\begin{equation}
\label{eq:step1-beta}
\max_{1\le i\le m}\mathrm{Dist}_{\mathcal H}(\hat f^{(i)},\mathcal S)\le \beta,
\qquad
\max_{1\le k\le K}\min_{1 \le i \le m}\|\hat f^{(i)}-g_k\|_{\mathcal H}\le \beta.
\end{equation}
Applying Lemma~\ref{lem:S4} with $s=0$, we obtain $d_{\max} \;\ge\; \sigma_\ast/2$.
Moreover, for any $k\neq \ell$, Lemma~\ref{lem:S1} with $\pi=e_k$ and $\tilde\pi=e_\ell$ gives $\|g_k-g_\ell\|_{\mathcal H}\ge \sqrt{\lambda_{K-1}(\tilde\Gamma)}\,\|e_k-e_\ell\|_2 = \sqrt{2}\,\sigma_\ast$, verifying the separation condition given in Equation~\eqref{eq:func-B22}.

We then apply Lemma~\ref{lem:S5}. Let $\mathcal V(\varepsilon_0,h_0)$ be defined as in Equation~\eqref{eq:V-eps-h}, with
\begin{equation}
\label{eq:step1-eps}
h_0 := \sigma_\ast/3,
\qquad
\varepsilon_0 := 15\max\{\sigma_\ast^{-1},\ \sigma_\ast^{-2}d_{\max}\}\beta .
\end{equation}
The assumption of the theorem yields $\varepsilon_0<1/2$.
Moreover, setting $t=7\beta/3$ gives $\varepsilon_0 \ge
6\sigma_\ast^{-1}\max\{1,d_{\max}/\sigma_\ast\}t$,
so the condition in Lemma~\ref{lem:S5} is satisfied. Applying Lemma~\ref{lem:S5} yields
\begin{equation}
\label{eq:step1-gap}
\max_{\hat f\in\mathcal S\setminus\mathcal V(\varepsilon_0,h_0)}
\| \hat f \|_{\mathcal H}
\;\le\;
d_{\max}-7\beta/3.
\end{equation}
Let $\mathcal K^\ast := \mathcal K (0)$. For any $k\in\mathcal K^\ast$, it follows from Equation~\eqref{eq:step1-beta} that there exists at least one index $i^\ast$ 
such that ${\|\hat f^{(i^\ast)}-g_k\|_{\mathcal H}\le \beta}$.
Since $\|g_k\|_{\mathcal H}=d_{\max}$ for $k\in\mathcal K^\ast$, the triangle
inequality implies
\[
\| \hat f^{(i^\ast)} \|_{\mathcal H}
\;\ge\;
\|g_k\|_{\mathcal H}-\beta
\;=\;
d_{\max}-\beta.
\]
Let $i_1:=\arg\max_i\|\hat f^{(i)}\|_{\mathcal H}$ be the index selected by SPA in the
first iteration. Then
\begin{equation}
\label{eq:step1-max}
\|\hat f^{(i_1)}\|_{\mathcal H}
\;\ge\;
\|\hat f^{(i^\ast)}\|_{\mathcal H}
\;\ge\;
d_{\max}-\beta.
\end{equation}
Combining Equations~\eqref{eq:step1-gap}~and~\eqref{eq:step1-max}, we conclude that
$\hat f^{(i_1)} \notin \mathcal S\setminus\mathcal V(\varepsilon_0,h_0)$. Equivalently,
\begin{equation}
\label{eq:step1-inclusion}
\hat f^{(i_1)} \in \mathcal V(\varepsilon_0,h_0)
\quad\text{or}\quad
\hat f^{(i_1)} \notin\mathcal S.
\end{equation}

Suppose first that $\hat f^{(i_1)} \in\mathcal V(\varepsilon_0,h_0)$. Since the sets
$\{\mathcal V_k(\varepsilon_0)\}_{k\in\mathcal K(h_0)}$ are disjoint, there exists
a unique $k_1\in\mathcal K(h_0)$ such that
$\hat f^{(i_1)} \in\mathcal V_{k_1}(\varepsilon_0)$. By Lemma~\ref{lem:S1},
\begin{equation}
\label{eq:step1-case1}
\|\hat f^{(i_1)} -g_{k_1}\|_{\mathcal H}
\;\le\;
2\gamma\,\varepsilon_0.
\end{equation}
Now suppose $\hat f^{(i_1)}\notin\mathcal S$. Let $\mathrm{proj}_{\mathcal S}(\hat f^{(i_1)})$
denote the closest point to $\hat f^{(i_1)}$ in $\mathcal S$. By Equation~\eqref{eq:step1-beta},
\[
\|\hat f^{(i_1)}-\mathrm{proj}_{\mathcal S}(\hat f^{(i_1)})\|_{\mathcal H}\le \beta.
\]
Combining this with Equation~\eqref{eq:step1-max} yields
\[
\|\mathrm{proj}_{\mathcal S}(\hat f^{(i_1)})\|_{\mathcal H}
\;\ge\;
\|\hat f^{(i_1)}\|_{\mathcal H}-\beta
\;\ge\;
d_{\max}-2\beta.
\]
Together with Equation~\eqref{eq:step1-gap}, this implies that
$\mathrm{proj}_{\mathcal S}(\hat f^{(i_1)})\in\mathcal V(\varepsilon_0,h_0)$.
Therefore, there exists a unique $k_1\in\mathcal K(h_0)$ such that
$\mathrm{proj}_{\mathcal S}(\hat f^{(i_1)})\in\mathcal V_{k_1}(\varepsilon_0)$.
By Lemma~\ref{lem:S1} again,
\begin{equation}
\label{eq:step1-case2}
\|\mathrm{proj}_{\mathcal S}(\hat f^{(i_1)})-g_{k_1}\|_{\mathcal H}
\;\le\;
2\gamma\,\varepsilon_0.
\end{equation}
Combining this with the triangle inequality gives
\[
\|\hat f^{(i_1)}-g_{k_1}\|_{\mathcal H}
\;\le\;
\|\hat f^{(i_1)}-\mathrm{proj}_{\mathcal S}(\hat f^{(i_1)})\|_{\mathcal H}
+
\|\mathrm{proj}_{\mathcal S}(\hat f^{(i_1)})-g_{k_1}\|_{\mathcal H}
\;\le\;
\beta+2\gamma\,\varepsilon_0.
\]

Putting the two cases together and substituting the value of $\varepsilon_0$
from Equation~\eqref{eq:step1-eps}, we conclude that
\begin{equation}
\label{eq:step1-final}
\|\hat f^{(i_1)}-g_{k_1}\|_{\mathcal H}
\;\le\;
\beta+2\gamma\,\varepsilon_0
\;\le\;
\Big(1+\frac{30\gamma}{\sigma_\ast}\max\{1,d_{\max}/\sigma_\ast\}\Big)\beta,
\end{equation}
for some $k_1\in\{1,\dots,K\}$. This shows that the function selected in the
first iteration of SPA lies in the neighborhood of a true basis function.

\paragraph{Step 2: Analysis of the general iteration.}

We proceed by induction. Suppose that after $(s-1)$ iterations where $s \in \{2, \ldots, K\}$, SPA has
selected indices $i_1,\dots,i_{s-1}$ such that
\[
\|\hat f^{(i_r)}-g_{k_r}\|_{\mathcal H}
\;\le\;
\Big(1+\frac{30\gamma}{\sigma_\ast}\max\{1,d_{\max}/\sigma_\ast\}\Big)\beta,
\qquad r=1,\dots,s-1,
\]
for distinct indices $k_1,\dots,k_{s-1}$.
Let $\mathcal H_{s-1}:=\mathrm{span}\{\hat f^{(i_1)},\dots,\hat f^{(i_{s-1})}\}$
and let $P_{s-1}$ be the orthogonal projection onto $\mathcal H_{s-1}$.
Define residuals
\[
\hat f^{(i)}_{(s)} := (I-P_{s-1})\hat f^{(i)},\qquad
g_{(s),k} := (I-P_{s-1})g_k.
\]
Let $\mathcal S_{(s)}:=\mathrm{conv}\{g_{(s),1},\dots,g_{(s),K}\}\subset\mathcal H$ and define the
$\varepsilon$-simplicial neighborhoods inside $\mathcal S_{(s)}$ by
\[
\mathcal V_{(s),k}(\varepsilon)
:=
\Big\{
\sum_{\ell=1}^K \pi_\ell g_{(s),\ell}:\ \pi\in\Delta_{K-1},\ \pi_k\ge 1-\varepsilon
\Big\},
\qquad 1\le k\le K.
\]
Let
\[
d_{\max}^{(s)} := \max_{\hat f\in \mathcal S_{(s)}} \|\hat f\|_{\mathcal H}
= \max_{1\le k\le K}\|g_{(s),k}\|_{\mathcal H},
\qquad
\mathcal K^{\ast}_{(s)} := \Big\{k:\ \|g_{(s),k}\|_{\mathcal H}=d_{\max}^{(s)}\Big\},
\]
and, for $h_0>0$,
\[
\mathcal K_{(s)}(h_0) := \Big\{k:\ \|g_{(s),k}\|_{\mathcal H}\ge d_{\max}^{(s)}-h_0\Big\},
\qquad
\mathcal V_{(s)}(\varepsilon_0,h_0):=\bigcup_{k\in \mathcal K_{(s)}(h_0)}\mathcal V_{(s),k}(\varepsilon_0).
\]
Since $I-P_{s-1}$ is an orthogonal projector and hence a contraction, $\|(I-P_{s-1})h_1-(I-P_{s-1})h_2\|_{\mathcal H}\le \|h_1-h_2\|_{\mathcal H}$ for all $h_1, h_2\in\mathcal H$.
Therefore, with $\beta=\beta(\hat f,G)$,
\begin{equation}
\label{eq:step2-beta-proj}
\max_{1\le i\le m}\mathrm{Dist}_{\mathcal H}(\hat f^{(i)}_{(s)},\mathcal S_{(s)})\le \beta,
\qquad
\max_{1\le k\le K}\min_{1 \le i \le m}\|\hat f^{(i)}_{(s)}-g_{(s),k}\|_{\mathcal H}\le \beta.
\end{equation}
Additionally, we have the following lemma. 
\begin{lemma}
\label{lem:S6}
Under the induction hypothesis and the assumptions of Theorem~\ref{thm:spa_pure}, we have
\begin{equation}
\label{eq:step2-proj-geom}
d_{\max}^{(s)} \ge \sigma_\ast/2,
\qquad
\min_{\substack{k\neq \ell\\ k,\ell\notin\{k_1,\dots,k_{s-1}\}}}
\|g_{(s),k}-g_{(s),\ell}\|_{\mathcal H} \ge \sqrt{2}\,\sigma_\ast,
\qquad
k_r\notin \mathcal K_{(s)}(h_0)\ \text{ for } r=1,\dots,s-1.
\end{equation}
\end{lemma}
We then apply Lemma~\ref{lem:S5} to the projected simplex. Similarly, as in Step~1, choose
\[
h_0=\sigma_\ast/3,
\qquad
\varepsilon_1=15\max\{\sigma_\ast^{-1},\ \sigma_\ast^{-2}d_{\max}^{(s)}\}\beta,
\qquad
t=7\beta/3.
\]
Note that $\varepsilon_1 \leq \varepsilon_0$ as $d_{\max}^{(s)} \leq d_{\max}$. We get
\begin{equation}
\label{eq:step2-gap-proj}
\max_{\hat f\in \mathcal S_{(s)}\setminus \mathcal V_{(s)}(\varepsilon_0,h_0)}
\|\hat f\|_{\mathcal H}
\;\le\;
\max_{\hat f\in \mathcal S_{(s)}\setminus \mathcal V_{(s)}(\varepsilon_1,h_0)}
\|\hat f\|_{\mathcal H}
\;\le\;
d_{\max}^{(s)}-7\beta/3.
\end{equation}
Let $i_s:=\arg\max_i \|\hat f^{(i)}_{(s)}\|_{\mathcal H}$ be the index selected by SPA in the $s$-th
iteration. By Equation~\eqref{eq:step2-beta-proj}, for any $k\in\mathcal K_{(s)}^{\ast}$ there exists
$i^\ast$ such that $\|\hat f^{(i^\ast)}_{(s)}-g_{(s),k}\|_{\mathcal H}\le \beta$, and hence
\[
\|\hat f^{(i^\ast)}_{(s)}\|_{\mathcal H}
\;\ge\;
\|g_{(s),k}\|_{\mathcal H}-\beta
\;=\;
d_{\max}^{(s)}-\beta.
\]
Since $\|\hat f^{(i_s)}_{(s)}\|_{\mathcal H}=\max_i\|\hat f^{(i)}_{(s)}\|_{\mathcal H}$, we obtain
\begin{equation}
\label{eq:step2-max-proj}
\|\hat f^{(i_s)}_{(s)}\|_{\mathcal H}\ge d_{\max}^{(s)}-\beta.
\end{equation}
Combining Equations~\eqref{eq:step2-gap-proj}~and~\eqref{eq:step2-max-proj}, we conclude that
$\hat f^{(i_s)}_{(s)}\notin \mathcal S_{(s)}\setminus \mathcal V_{(s)}(\varepsilon_0,h_0)$; namely,
\begin{equation}
\label{eq:step2-inclusion-proj}
\hat f^{(i_s)}_{(s)}\in \mathcal V_{(s)}(\varepsilon_0,h_0)
\quad\text{or}\quad
\hat f^{(i_s)}_{(s)}\notin \mathcal S_{(s)}.
\end{equation}

In both cases, we route through the projection onto $\mathcal S$.
By Equation~\eqref{eq:step2-beta-proj}, $\mathrm{Dist}_{\mathcal H}(\hat f^{(i_s)},\mathcal S)\le \beta$, so
\[
\|\hat f^{(i_s)}-\mathrm{proj}_{\mathcal S}(\hat f^{(i_s)})\|_{\mathcal H}\le \beta.
\]
First, consider $\hat f^{(i_s)}_{(s)}\in \mathcal V_{(s)}(\varepsilon_0,h_0)$. Since $\{\mathcal V_{(s),k}(\varepsilon_0)\}_{k\in\mathcal K_{(s)}(h_0)}$ are disjoint, there exists a unique $k_s\in\mathcal K_{(s)}(h_0)$ such that
$\hat f^{(i_s)}_{(s)}\in \mathcal V_{(s),k_s}(\varepsilon_0)$, meaning $\hat f^{(i_s)}_{(s)} = \sum_{\ell=1}^K \pi_\ell g_{(s),\ell}$ with $\pi_{k_s}\ge 1-\varepsilon_0$. Since $(I-P_{s-1})$ does not change the weights, $\mathrm{proj}_{\mathcal S}(\hat f^{(i_s)})$ lies in $\mathcal V_{k_s}(\varepsilon_0)$. By Lemma~\ref{lem:S1},
\begin{equation}
\label{eq:step2-case1-proj}
\|\mathrm{proj}_{\mathcal S}(\hat f^{(i_s)})-g_{k_s}\|_{\mathcal H}\le 2\gamma\,\varepsilon_0.
\end{equation}
Now consider $\hat f^{(i_s)}_{(s)}\notin \mathcal S_{(s)}$.
Let $\mathrm{proj}_{\mathcal S_{(s)}}(\hat f^{(i_s)}_{(s)})$ denote the closest point to
$\hat f^{(i_s)}_{(s)}$ in $\mathcal S_{(s)}$.
By Equation~\eqref{eq:step2-beta-proj},
\[
\|\hat f^{(i_s)}_{(s)}-\mathrm{proj}_{\mathcal S_{(s)}}(\hat f^{(i_s)}_{(s)})\|_{\mathcal H}\le \beta.
\]
Combining this with Equation~\eqref{eq:step2-max-proj} yields
\[
\|\mathrm{proj}_{\mathcal S_{(s)}}(\hat f^{(i_s)}_{(s)})\|_{\mathcal H}
\ge
\|\hat f^{(i_s)}_{(s)}\|_{\mathcal H}-\beta
\ge
d_{\max}^{(s)}-2\beta.
\]
Together with Equation~\eqref{eq:step2-gap-proj}, this implies
$\mathrm{proj}_{\mathcal S_{(s)}}(\hat f^{(i_s)}_{(s)})\in \mathcal V_{(s)}(\varepsilon_0,h_0)$.
Therefore, there exists a unique $k_s\in\mathcal K_{(s)}(h_0)$ such that
$\mathrm{proj}_{\mathcal S_{(s)}}(\hat f^{(i_s)}_{(s)})\in \mathcal V_{(s),k_s}(\varepsilon_0)$,
which implies $\mathrm{proj}_{\mathcal S}(\hat f^{(i_s)})\in \mathcal V_{k_s}(\varepsilon_0)$.
By Lemma~\ref{lem:S1},
\begin{equation}
\label{eq:step2-case2-proj}
\|\mathrm{proj}_{\mathcal S}(\hat f^{(i_s)})-g_{k_s}\|_{\mathcal H}\le 2\gamma\,\varepsilon_0.
\end{equation}
In both cases, combining with the triangle inequality gives
\[
\|\hat f^{(i_s)}-g_{k_s}\|_{\mathcal H}
\le
\|\hat f^{(i_s)}-\mathrm{proj}_{\mathcal S}(\hat f^{(i_s)})\|_{\mathcal H}
+
\|\mathrm{proj}_{\mathcal S}(\hat f^{(i_s)})-g_{k_s}\|_{\mathcal H}
\le
\beta+2\gamma\,\varepsilon_0.
\]
Therefore,
\begin{equation}
\label{eq:step2-final-proj}
\|\hat f^{(i_s)}-g_{k_s}\|_{\mathcal H}
\le
\beta+2\gamma\,\varepsilon_0
\le
\Big(1+\frac{30\gamma}{\sigma_\ast}\max\{1,d_{\max}/\sigma_\ast\}\Big)\beta.
\end{equation}
In particular, $k_s\notin\{k_1,\dots,k_{s-1}\}$ by Equation~\eqref{eq:step2-proj-geom}.
This completes the induction.
\qed

\subsubsection{Proof of Theorem~\ref{thm:bcons}}
\label{sec:denoising_proof}

Recall from Assumption~\ref{cond:purity} the near-pure count $M_k(\delta_m) := |\{1\le i\le m:\, \pi_{ik}\ge 1-\delta_m\}|$ and that $\min_{1\le k\le K} M_k(\delta_m)\ge 1$.
We specify the denoising parameters $(N,\Delta)$ as follows:
\begin{eqnarray}
\label{eq:denoise_params}
\Delta = C\big(\delta_m\,\gamma(G) + n_{\min}^{(a-r)/2}\big), \qquad
N = \lfloor c\, \min_{1\le k\le K} M_k(\delta_m) \rfloor,
\end{eqnarray}
for a sufficiently large constant $C\ge 5$ and a constant $0 < c \le 1$.
The neighborhood radius $\Delta$ is proportional to the sum of the purity error and the estimation error, and vanishes under Assumptions~\ref{cond:est_error}~and~\ref{cond:purity}. The pruning threshold $N$ is set as a fraction of the smallest near-pure count across all vertices, ensuring that near-vertex studies are not pruned.
When $\min_k M_k(\delta_m) = 1$ (only one near-pure study per vertex), we have $N=1$, so no function is pruned and only the averaging step is applied.
When $\min_k M_k(\delta_m) \gg 1$ (many near-pure studies per vertex), the pruning threshold is larger and the denoising step can remove outlying functions.

Let $\hat f^{*,(1)},\ldots,\hat f^{*,(m')}$ denote the denoised functions output by the denoising step of Algorithm~\ref{alg:spa}, where $m'\le m$ is the number of retained functions. Write $\gamma:=\gamma(G)$, $d_{\max}:=d_{\max}(G)$, and $\varepsilon_n := \max_{1\le i\le m}\|\epsilon^{(i)}\|_{\mathcal H}$. By Lemma~\ref{lem:beta_control}, $\varepsilon_n = o_P(1)$.
The proof shows that $\beta(\hat f^{*},G)=o_P(1)$ after the denoising step, verifies the condition of Theorem~\ref{thm:spa_pure}, and then applies the finite-sample error bound.

\paragraph{Step 1: Distance to $\mathcal S$ is preserved after averaging.}
For each retained function $\hat f^{*,(i)}$, let $\mathcal N_i:=\{j:\|\hat f^{(j)}-\hat f^{(i)}\|_{\mathcal H}\le \Delta\}$ be its $\Delta$-neighborhood. By construction, $\hat f^{*,(i)} = |\mathcal N_i|^{-1}\sum_{j\in\mathcal N_i}\hat f^{(j)}$.
For each $j$, since $f^{(j)}\in\mathcal S$, we have $\mathrm{Dist}_{\mathcal H}(\hat f^{(j)},\mathcal S)\le \|\epsilon^{(j)}\|_{\mathcal H}$.
Let $s_j:=\mathrm{proj}_{\mathcal S}(\hat f^{(j)})$ for each $j\in\mathcal N_i$. Since $\mathcal S$ is convex, $\bar s := |\mathcal N_i|^{-1}\sum_{j\in\mathcal N_i}s_j\in\mathcal S$.
Therefore,
\[
\mathrm{Dist}_{\mathcal H}(\hat f^{*,(i)},\mathcal S)
\;\le\;
\|\hat f^{*,(i)}-\bar s\|_{\mathcal H}
\;\le\;
\frac{1}{|\mathcal N_i|}\sum_{j\in\mathcal N_i}\|\hat f^{(j)}-s_j\|_{\mathcal H}
\;\le\;
\varepsilon_n.
\]

\paragraph{Step 2: Near-vertex studies survive pruning.}
Fix $k\in\{1,\ldots,K\}$ and let $I_k:=\{i:\pi_{ik}\ge 1-\delta_m\}$, so $|I_k| = M_k(\delta_m) \ge N/c \ge N$.
For any $i,j\in I_k$, since $\pi_{ik},\pi_{jk}\ge 1-\delta_m$, we have $\sum_{\ell\neq k}\pi_{i\ell}\le\delta_m$ and $\sum_{\ell\neq k}\pi_{j\ell}\le\delta_m$. Hence $\|\bpi_i-\bpi_j\|_1\le 4\delta_m$, and by Lemma~\ref{lem:S1},
\[
\|f^{(i)}-f^{(j)}\|_{\mathcal H}
\;=\;
\|G\bpi_i - G\bpi_j\|_{\mathcal H}
\;\le\;
\gamma\,\|\bpi_i-\bpi_j\|_1
\;\le\;
4\,\gamma\,\delta_m.
\]
By the triangle inequality,
$\|\hat f^{(i)}-\hat f^{(j)}\|_{\mathcal H}
\le
4\,\gamma\,\delta_m + 2\varepsilon_n$.
With the choice of $\Delta$ in Equation~\eqref{eq:denoise_params}, $4\gamma\delta_m + 2\varepsilon_n \le \Delta$ with probability tending to one.
Therefore, all studies in $I_k$ are mutual $\Delta$-neighbors, and each has $|\mathcal N_i|\ge |I_k| \ge N$, so none is pruned.

\paragraph{Step 3: Denoised near-vertex functions are close to vertices.}
Let $i_k\in I_k$ be the index achieving $\min_{i}(1-\pi_{ik})$. Since $i_k$ is retained (by Step~2) and $\|\hat f^{(j)}-\hat f^{(i_k)}\|_{\mathcal H}\le \Delta$ for all $j\in\mathcal N_{i_k}$, we have for each such $j$,
\[
\|\hat f^{(j)}-g_k\|_{\mathcal H}
\;\le\;
\|\hat f^{(j)}-\hat f^{(i_k)}\|_{\mathcal H}+\|\hat f^{(i_k)}-g_k\|_{\mathcal H}
\;\le\;
\Delta + 2\,\delta_m\, d_{\max} + \varepsilon_n.
\]
Since $\hat f^{*,(i_k)}$ is the average over $\mathcal N_{i_k}$, the triangle inequality yields
\[
\|\hat f^{*,(i_k)}-g_k\|_{\mathcal H}
\;\le\;
\Delta + 2\,\delta_m\, d_{\max} + \varepsilon_n.
\]

\paragraph{Step 4: Controlling $\beta$ and applying Theorem~\ref{thm:spa_pure}.}
Steps~1 and~3 give
\[
\beta(\hat f^{*},G)
\;\le\;
\Delta + 2\,\delta_m\,d_{\max} + \varepsilon_n
\;=\;
O_P\!\big(\delta_m\,\gamma + n_{\min}^{(a-r)/2} + \delta_m\,d_{\max}\big)
\;=\;
o_P(1),
\]
where the last equality uses Assumptions~\ref{cond:est_error} and~\ref{cond:purity}.
By the non-degeneracy condition in Assumption~\ref{cond:low_rank}, $\sigma_\ast>0$. Since the basis functions $\{g_k\}_{k=1}^K$ are fixed, $\sigma_\ast$, $d_{\max}$, and $\gamma$ are positive constants that do not depend on $m$ or $n_i$. Therefore, the condition \eqref{eq:thmS1.1} of Theorem~\ref{thm:spa_pure} is satisfied with probability tending to one. On this event, Theorem~\ref{thm:spa_pure} applied to the denoised functions $\hat f^{*,(1)},\ldots,\hat f^{*,(m')}$ yields, up to permutation,
\[
\max_{1\le k\le K}\|\hat g_k-g_k\|_{\mathcal H}
\;\le\;
\underbrace{\Bigg(1+\frac{30\,\gamma}{\sigma_\ast}\max\Big\{1,\,\frac{d_{\max}}{\sigma_\ast}\Big\}\Bigg)}_{=:\,C_0}
\beta(\hat f^{*},G)
\;=\;
o_P(1),
\]
since $C_0$ is a fixed constant depending only on the basis functions.

It remains to extend the conclusion from the case $\khat = K$ to general $\khat$. Since d-fSPA is a greedy iterative algorithm, the first $j$ selections depend only on $j$ and not on the total number of iterations. Therefore, running SPA for $\khat$ iterations produces the same first $\min(K, \khat)$ selections as running for $K$ iterations. Combined with the bound for the $K$-iteration case, this yields
\[
\max_{1\le k\le \min(K,\khat)}\|\hat g_k-g_k\|_{\mathcal H} \;\le\; C_0\,\beta(\hat f^{*},G) \;=\; o_P(1).
\]
\qed

\subsubsection{Proof of Lemmas}
\begin{proof}
[Proof of Lemma~\ref{lem:S1}] 
Fix $\pi,\tilde\pi\in\Delta_{K-1}$ and set $a := \pi-\tilde\pi\in\mathbb R^K$. Since
$\mathbf 1^\top a = \sum_{k=1}^K(\pi_k-\tilde\pi_k)=0$, we have
\[
Ga \;=\; \sum_{k=1}^K a_k g_k \;=\; \sum_{k=1}^K a_k \tilde g_k \;=\; \tilde G a.
\]
Therefore,
\[
\|G\pi-G\tilde\pi\|_{\mathcal H}^2
\;=\;\|\tilde G a\|_{\mathcal H}^2
\;=\;\langle \tilde G a,\tilde G a\rangle
\;=\;\langle a,\tilde G^\ast \tilde G a\rangle
\;=\; a^\top \tilde\Gamma\, a.
\]
Since $\tilde\Gamma$ is symmetric positive semidefinite and $\mathrm{rank}(\tilde\Gamma)\le K-1$, the Rayleigh quotient bound yields
\[
a^\top \tilde\Gamma\, a
\;\ge\;
\lambda_{K-1}(\tilde\Gamma)\,\|a\|_2^2,
\]
where $\lambda_{K-1}(\tilde\Gamma)$ denotes the smallest nonzero eigenvalue (equivalently, the
$(K-1)$-th largest eigenvalue) of $\tilde\Gamma$. Taking square roots gives the lower bound.

For the upper bound in Equation~\eqref{eq:lemmaS1.1}, fix any $g_0\in\mathcal S$. Since $\sum_{k=1}^K a_k = 0$, we have $\sum_{k=1}^K a_k g_k = \sum_{k=1}^K a_k (g_k - g_0)$. By the triangle inequality,
\[
\|G\pi-G\tilde\pi\|_{\mathcal H}
=\biggl\|\sum_{k=1}^K a_k (g_k - g_0)\biggr\|_{\mathcal H}
\le \sum_{k=1}^K |a_k|\,\|g_k - g_0\|_{\mathcal H}
\le \max_{1\le k\le K}\|g_k - g_0\|_{\mathcal H} \cdot \|\pi-\tilde\pi\|_1.
\]
Minimizing the right-hand side over $g_0\in\mathcal S$ yields $\gamma(G)\,\|\pi-\tilde\pi\|_1$.

We now prove Equation~\eqref{eq:lemmaS1.2}. Fix $1\le s\le K-2$ and suppose $\pi$ and $\tilde\pi$ share at least
$s$ common entries. Let $I\subset\{1,\dots,K\}$ be a set of indices with $|I|=s$ such that
$\pi_k=\tilde\pi_k$ for all $k\in I$, and let $J:=I^c$ so $|J|=K-s$. Then $a_I=0$ and hence
\[
\|G\pi-G\tilde\pi\|_{\mathcal H}^2
= a^\top \tilde\Gamma a
= a_J^\top \tilde\Gamma_{JJ}\, a_J,
\]
where $\tilde\Gamma_{JJ}$ is the principal submatrix of $\tilde\Gamma$ indexed by $J$.
Since $\mathbf 1^\top a = 0$ and $a_I = 0$, we have $\mathbf 1_J^\top a_J = 0$, so $a_J$ lies in the $(K-s-1)$-dimensional subspace $\{v\in\mathbb R^{K-s}: \mathbf 1_{J}^\top v=0\}$. By the Rayleigh quotient bound restricted to this subspace,
\[
a_J^\top \tilde\Gamma_{JJ}\, a_J
\;\ge\;
\lambda_{K-s-1}(\tilde\Gamma_{JJ})\,\|a_J\|_2^2,
\]
where $\lambda_{K-s-1}(\tilde\Gamma_{JJ})$ denotes the $(K-s-1)$-th largest eigenvalue of the $(K-s)\times(K-s)$ matrix $\tilde\Gamma_{JJ}$.
By the Cauchy interlacing theorem for eigenvalues of a symmetric matrix and its principal submatrices, $\lambda_{K-s-1}(\tilde\Gamma_{JJ}) \ge \lambda_{K-1}(\tilde\Gamma) = \sigma_\ast^2$.
In particular, $\lambda_{K-s-1}(\tilde\Gamma_{JJ}) \ge \lambda_{K-1-s}(\tilde\Gamma)$.
This, together with $\|a_J\|_2=\|a\|_2=\|\pi-\tilde\pi\|_2$, yields
\[
\|G\pi-G\tilde\pi\|_{\mathcal H}
\;\ge\;
\sqrt{\lambda_{K-1-s}(\tilde\Gamma)}\,\|\pi-\tilde\pi\|_2.
\]
\end{proof}

\begin{proof}[Proof of Lemma~\ref{lem:S2}]
Write
\[
f \;:=\; \sum_{k=1}^K \pi_k g_k \in \mathcal H,
\qquad
L \;:=\; \sum_{k=1}^K \pi_k \|g_k\|_{\mathcal H}.
\]
By the triangle inequality, $\| f\|_{\mathcal H}\le L$. We aim to lower bound
$L-\|f\|_{\mathcal H}$.

First expand $\|f\|_{\mathcal H}^2$ using bilinearity of the inner product:
\begin{equation}
\label{eq:lem2-expand}
\|f\|_{\mathcal H}^2
= \Big\langle \sum_{i=1}^K \pi_i g_i,\ \sum_{j=1}^K \pi_j g_j \Big\rangle
= \sum_{i=1}^K \pi_i^2 \|g_i\|_{\mathcal H}^2 +
\sum_{i\neq j} \pi_i\pi_j \langle g_i,g _j\rangle.
\end{equation}
Next, we use the following identity, valid in any inner product space: for any $u,v\in\mathcal H$,
\begin{equation}
\label{eq:polar-identity}
2\langle u,v\rangle
=
2\|u\|_{\mathcal H}\|v\|_{\mathcal H}
+
(\|u\|_{\mathcal H}-\|v\|_{\mathcal H})^2
-
\|u-v\|_{\mathcal H}^2.
\end{equation}
Applying Equation~\eqref{eq:polar-identity} with $u=g_i$ and $v=g_j$, and using the definitions of
$a$ and $b$, we have for all $i\neq j$,
\[
(\| g_i\|_{\mathcal H}-\|g_j\|_{\mathcal H})^2 \le b^2,
\qquad
\|g_i-g_j\|_{\mathcal H}^2 \ge a^2,
\]
hence
\begin{equation}
\label{eq:cross-bound}
\langle g_i,g_j\rangle
\;\le\;
\|g_i\|_{\mathcal H}\|g_j\|_{\mathcal H} - \frac{a^2-b^2}{2}.
\end{equation}
Plugging Equation~\eqref{eq:cross-bound} into Equation~\eqref{eq:lem2-expand} yields
\begin{align}
\|f\|_{\mathcal H}^2
&\le
\sum_{i=1}^K \pi_i^2 \|g_i\|_{\mathcal H}^2
+
\sum_{i\neq j} \pi_i\pi_j
\Big(\|g_i\|_{\mathcal H}\|g_j\|_{\mathcal H} - \frac{a^2-b^2}{2}\Big) \notag\\
&=
\Big(\sum_{i=1}^K \pi_i \|g_i\|_{\mathcal H}\Big)^2
\;-\;
\frac{a^2-b^2}{2}\sum_{i\neq j}\pi_i\pi_j \notag\\
&=
L^2 \;-\; \frac{a^2-b^2}{2}\sum_{i\neq j}\pi_i\pi_j.
\label{eq:lem2-L2}
\end{align}
Since $\sum_{j\neq i}\pi_j = 1-\pi_i$, we have
\[
\sum_{i\neq j}\pi_i\pi_j
=
\sum_{i=1}^K \pi_i \sum_{j\neq i}\pi_j
=
\sum_{i=1}^K \pi_i(1-\pi_i).
\]
Combining this with Equation~\eqref{eq:lem2-L2} gives
\begin{equation}
\label{eq:lem2-L2-2}
\|f\|_{\mathcal H}^2
\le
L^2 \;-\; \frac{a^2-b^2}{2}\sum_{i=1}^K \pi_i(1-\pi_i).
\end{equation}
Finally, note that $L+\| f\|_{\mathcal H}\le 2L$ (since $\| f\|_{\mathcal H}\le L$). Therefore,
\[
L-\|f\|_{\mathcal H}
=
\frac{L^2-\| f\|_{\mathcal H}^2}{L+\| f\|_{\mathcal H}}
\ge
\frac{L^2-\|f\|_{\mathcal H}^2}{2L}.
\]
Using Equation~\eqref{eq:lem2-L2-2} in the numerator yields
\[
L-\| f\|_{\mathcal H}
\ge
\frac{a^2-b^2}{4L}\sum_{i=1}^K \pi_i(1-\pi_i).
\]
\end{proof}

\begin{proof}[Proof of Lemma~\ref{lem:S3}]
Write $\Gamma_P \;=\; G^{\ast}G \;-\; G^{\ast}PG \;=\; \Gamma - B$ with $B:=G^{\ast}PG$.
Since $P$ is an orthogonal projection, it is positive semi-definite. Hence for any $a\in\mathbb R^K$,
\[
a^\top B a \;=\; \langle PGa,\,Ga\rangle_{\mathcal H}
\;=\; \langle PGa,\,PGa\rangle_{\mathcal H}
\;=\; \|PGa\|_{\mathcal H}^2 \;\ge\;0,
\]
so $B\succeq 0$. Moreover, $\mathrm{rank}(B)\le \mathrm{rank}(P)=s$.
A standard eigenvalue inequality for a rank-$s$ positive semi-definite perturbation (e.g., Weyl-type interlacing for $A-B$ with $\mathrm{rank}(B)\le s$) yields
\[
\lambda_{K-1-s}(\Gamma_P)
\;=\; \lambda_{K-1-s}(\Gamma-B)
\;\ge\; \lambda_{K-1}(\Gamma).
\]
\end{proof}

\begin{proof}[Proof of Lemma~\ref{lem:S4}]
Write $\ell_{\max}:=\max_{1\le k\le K}\|g_k\|_{\mathcal H}$. By assumption, there exist
indices $\{k_1,\dots,k_s\}$ such that $\|g_{k_j}\|_{\mathcal H}\le \delta$ for $j=1,\dots,s$. Hence
\begin{equation}
\label{eq:lem4-trace-upper}
\sum_{k=1}^K \|g_k\|_{\mathcal H}^2
\;\le\;
s\delta^2 + (K-s)\ell_{\max}^2 .
\end{equation}
On the other hand,
\[
\sum_{k=1}^K \|g_k\|_{\mathcal H}^2
\;=\;
\sum_{k=1}^K \langle g_k,g_k\rangle
\;=\;
\mathrm{trace}(\Gamma)
\;=\;
\sum_{j=1}^K \lambda_j(\Gamma).
\]
Since $\lambda_1(\Gamma)\ge\cdots\ge\lambda_K(\Gamma)\ge 0$, we have
\begin{equation}
\label{eq:lem4-trace-lower}
\mathrm{trace}(\Gamma)
\;=\;
\sum_{j=1}^K \lambda_j(\Gamma)
\;\ge\;
\sum_{j=1}^{K-1-s} \lambda_j(\Gamma)
\;\ge\;
(K-1-s)\lambda_{K-1-s}(\Gamma).
\end{equation}
Combining Equations~\eqref{eq:lem4-trace-upper}~and~\eqref{eq:lem4-trace-lower} yields
\[
s\delta^2 + (K-s)\ell_{\max}^2
\;\ge\;
(K-1-s)\lambda_{K-1-s}(\Gamma),
\]
and therefore
\begin{equation}
\label{eq:lem4-lmax-basic}
\ell_{\max}^2
\;\ge\;
\frac{(K-s-1)\lambda_{K-1-s}(\Gamma) - s\delta^2}{K-s}.
\end{equation}
Recall that $\lambda_{K-1-s}(\Gamma)\ge 2(K-2)\delta^2$.
Since the function $s/(K-s-1)$ is increasing in $s$ over $1\le s\le K-2$, we have
\[
\frac{s}{K-s-1}\;\le\; K-2
\quad\Longrightarrow\quad
2\frac{s}{K-s-1}\delta^2 \;\le\; 2(K-2)\delta^2 \;\le\; \lambda_{K-1-s}(\Gamma).
\]
Equivalently,
\[
s\delta^2 \;\le\; \frac{K-s-1}{2}\lambda_{K-1-s}(\Gamma).
\]
Substituting this bound into Equation~\eqref{eq:lem4-lmax-basic} gives
\[
\ell_{\max}^2
\;\ge\;
\frac{(K-s-1)\lambda_{K-1-s}(\Gamma) - \frac{K-s-1}{2}\lambda_{K-1-s}(\Gamma)}{K-s}
\;=\;
\frac{K-s-1}{2(K-s)}\lambda_{K-1-s}(\Gamma).
\]
Taking square roots,
\[
\ell_{\max}
\;\ge\;
\sqrt{\frac{K-s-1}{2(K-s)}}\;\sqrt{\lambda_{K-1-s}(\Gamma)}.
\]
Finally, since $1\le K-s-1\le K-s$ for $0\le s\le K-2$, we have
$\sqrt{(K-s-1)/(2(K-s))}\ge 1/2$, which implies
\[
\ell_{\max}\;\ge\;\frac12\sqrt{\lambda_{K-1-s}(\Gamma)}.
\]
\end{proof}

\begin{proof}
[Proof of Lemma~\ref{lem:S5}]
Write $\mathcal K=\mathcal K(h_0)$, $\mathcal V_k=\mathcal V_k(\varepsilon_0)$, and $\mathcal V=\mathcal V(\varepsilon_0,h_0)$ for short. By definition of $\mathcal K$,
\begin{equation}
\label{eq:B55-func}
d_{\max}-h_0 \le \|g_k\|_{\mathcal H}\le d_{\max}\ \ \text{for }k\in\mathcal K,
\qquad
\|g_k\|_{\mathcal H}\le d_{\max}-h_0\ \ \text{for }k\notin\mathcal K.
\end{equation}

Fix any $f\in \mathcal S\setminus \mathcal V$. There exist $\pi\in\Delta_{K-1}$ such that $f=\sum_{k=1}^K \pi_k g_k$. Since $f\notin \mathcal V=\cup_{k\in\mathcal K}\mathcal V_k(\varepsilon_0)$, we have
\begin{equation}
\label{eq:B56-func}
\max_{k\in\mathcal K}\ \pi_k \le 1-\varepsilon_0.
\end{equation}
Define
\begin{equation}
\label{eq:B57-func}
\rho := \sum_{k\in\mathcal K}\pi_k,
\qquad
\eta :=
\begin{cases}
\rho^{-1}\sum_{k\in\mathcal K}\pi_k g_k, & \rho\neq 0,\\
0, & \rho=0.
\end{cases}
\end{equation}
Then $f$ can be written as $f = \rho \eta + \sum_{k\notin\mathcal K}\pi_k g_k$, and by the triangle inequality together with Equation~\eqref{eq:B55-func},
\begin{equation}
\label{eq:B58-func}
\|f\|_{\mathcal H}
\le
\rho \|\eta\|_{\mathcal H} + \sum_{k\notin\mathcal K}\pi_k \|g_k\|_{\mathcal H}
\le
\rho\|\eta\|_{\mathcal H} + (1-\rho)(d_{\max}-h_0).
\end{equation}
We then consider two cases.

\noindent\emph{Case 1: $1-\rho \ge \varepsilon_0/2$.}
Since $\|\eta\|_{\mathcal H}\le d_{\max}$ (as $\eta\in \mathcal S$ when $\rho>0$), Equation~\eqref{eq:B58-func} gives
\begin{equation}
\label{eq:B59-func}
\|f\|_{\mathcal H}
\le
\rho d_{\max} + (1-\rho)(d_{\max}-h_0)
=
d_{\max}-(1-\rho)h_0
\le
d_{\max}-\frac{h_0\varepsilon_0}{2}.
\end{equation}

\smallskip
\noindent\emph{Case 2: $1-\rho < \varepsilon_0/2$.}
We first claim that $|\mathcal K|\ge 2$. If $\mathcal K=\{k^\ast\}$ is a singleton, then $\rho=\pi_{k^\ast}$ and Equation~\eqref{eq:B56-func} yields $\rho\le 1-\varepsilon_0$, hence $1-\rho\ge \varepsilon_0$, contradicting $1-\rho<\varepsilon_0/2$. Thus, $\eta$ is a convex combination of at least two points in $\{g_k:k\in\mathcal K\}$.
Apply Lemma~\ref{lem:S2} to the family $\{g_k:k\in\mathcal K\}$ with weights
$w_k:=\pi_k/\rho$ (so $\sum_{k\in\mathcal K}w_k=1$) and with
\[
a := \min_{\substack{k\neq \ell\\ k,\ell\in\mathcal K}}\|g_k-g_\ell\|_{\mathcal H}
\ \ge\ \sqrt{2}\,\sigma_\ast,
\qquad
b := \max_{\substack{k\neq \ell\\ k,\ell\in\mathcal K}}\big|\|g_k\|_{\mathcal H}-\|g_\ell\|_{\mathcal H}\big|
\ \le\ h_0,
\]
where we used the separation assumption restricted to vertices in $\mathcal K(h_0)$ and Equation~\eqref{eq:B55-func}. (Note that the proof only requires the separation condition among vertices in $\mathcal K(h_0)$, not among all $K$ vertices.)
Let
$L:=\sum_{k\in\mathcal K} w_k \|g_k\|_{\mathcal H}\le d_{\max}$. Lemma~\ref{lem:S2} yields
\begin{equation}
\label{eq:B61-func}
\|\eta\|_{\mathcal H}
\le
L - \frac{a^2-b^2}{4L}\sum_{k\in\mathcal K} w_k(1-w_k)
\le
d_{\max}-\frac{2\sigma_\ast^2-h_0^2}{4d_{\max}}
\sum_{k\in\mathcal K} w_k(1-w_k).
\end{equation}
Since $w_k=\pi_k/\rho$, we have $\sum_{k\in\mathcal K} w_k(1-w_k)=\sum_{k\in\mathcal K} w_k(1-\rho^{-1}\pi_k)$.
Moreover, for each $k\in\mathcal K$, using Equation~\eqref{eq:B56-func} and $1-\rho=\sum_{j\notin\mathcal K}\pi_j$,
\[
1-\rho^{-1}\pi_k
=
\rho^{-1}(1-\pi_k)-\rho^{-1}(1-\rho)
\ge
\rho^{-1}\big(\varepsilon_0-(1-\rho)\big).
\]
Plugging this into Equation~\eqref{eq:B61-func} and using $\sum_{k\in\mathcal K}w_k=1$ gives
\begin{equation}
\label{eq:B62-func}
\|\eta\|_{\mathcal H}
\le
d_{\max}
-\frac{(2\sigma_\ast^2-h_0^2)\big(\varepsilon_0-(1-\rho)\big)}{4\rho d_{\max}}.
\end{equation}
Since $1-\rho<\varepsilon_0/2$, we have $\varepsilon_0-(1-\rho)\ge \varepsilon_0/2$, hence
\[
\|\eta\|_{\mathcal H}
\le
d_{\max} - \frac{(2\sigma_\ast^2-h_0^2)\varepsilon_0}{8\rho d_{\max}}.
\]
Substitute this bound into Equation~\eqref{eq:B58-func} and use $1-\rho\ge 0$:
\begin{align}
\label{eq:B63-func}
\|f\|_{\mathcal H}
&\le
\rho\Big(d_{\max}-\frac{(2\sigma_\ast^2-h_0^2)\varepsilon_0}{8\rho d_{\max}}\Big)
+(1-\rho)(d_{\max}-h_0)\nonumber\\
&\le
\rho\Big(d_{\max}-\frac{(2\sigma_\ast^2-h_0^2)\varepsilon_0}{8\rho d_{\max}}\Big)
+(1-\rho)d_{\max}
=
d_{\max}-\frac{(2\sigma_\ast^2-h_0^2)\varepsilon_0}{8d_{\max}}.
\end{align}

Combining Equations~\eqref{eq:B59-func}~(Case 1)~and~\eqref{eq:B63-func} (Case 2), and setting $h_0=\sigma_\ast/3$,
we obtain for all $f\in\mathcal S\setminus \mathcal V(\varepsilon_0,h_0)$,
\[
\|f\|_{\mathcal H}
\le
d_{\max}-\min\Big\{\frac{\sigma_\ast}{6},\ \frac{17\sigma_\ast^2}{72 d_{\max}}\Big\}\varepsilon_0
\le
d_{\max}-\min\Big\{\frac{\sigma_\ast}{6},\ \frac{\sigma_\ast^2}{6 d_{\max}}\Big\}\varepsilon_0,
\]
where the second inequality uses $17/72 > 1/6$.
Therefore, a sufficient condition for $\|f\|_{\mathcal H}\le d_{\max}-t$ is
\[
\min\Big\{\frac{\sigma_\ast}{6},\ \frac{\sigma_\ast^2}{6d_{\max}}\Big\}\varepsilon_0 \ge t
\quad\Longleftrightarrow\quad
\varepsilon_0 \ge 6\sigma_\ast^{-1}\max\Big\{1,\frac{d_{\max}}{\sigma_\ast}\Big\}t.
\]
\end{proof}

\begin{proof}
[Proof of Lemma~\ref{lem:S6}]
Fix an iteration index $s\in\{2,\dots,K\}$ and adopt the notation in Step~2 of the proof of Theorem~\ref{thm:spa_pure}.
Define the projected operator $G_{(s)}: \mathbb R^K\to\mathcal H$ by $G_{(s)}:= (I-P_{s-1})G$, so $G_{(s)}e_k=g_{(s),k}$.
Let the corresponding Gram matrix be $\Gamma_{(s)}:=G_{(s)}^\ast G_{(s)}$ and let $\widetilde\Gamma_{(s)}:=\widetilde G_{(s)}^\ast \widetilde G_{(s)}$ be the centered Gram matrix, where $\widetilde G_{(s)}:=(I-P_{s-1})\widetilde G$.

Since $P_{s-1}$ is an orthogonal projection with $\mathrm{rank}(P_{s-1})=s-1$, Lemma~\ref{lem:S3} 
$\widetilde G$ with $P=P_{s-1}$ and with $s$ replaced by $s-1$) 
yields
\begin{equation}
\label{eq:lem6-sv}
\lambda_{K-1-(s-1)}\big(\widetilde\Gamma_{(s)}\big)
\;\ge\;
\lambda_{K-1}(\widetilde\Gamma)
\;=\;
\sigma_\ast^2.
\end{equation}
Fix $k\neq \ell$ with $k,\ell\notin\{k_1,\dots,k_{s-1}\}$. Consider $\pi=e_k$ and $\tilde\pi=e_\ell$. They share at least $(s-1)$ common entries (indeed, they share zeros on all indices in $\{k_1,\dots,k_{s-1}\}$), so Lemma~\ref{lem:S1}  applied to the mapping
$\pi\mapsto G_{(s)}\pi$ gives
\[
\|g_{(s),k}-g_{(s),\ell}\|_{\mathcal H}
=
\|G_{(s)}e_k-G_{(s)}e_\ell\|_{\mathcal H}
\ge
\sqrt{\lambda_{K-1-(s-1)}\big(\widetilde\Gamma_{(s)}\big)}\ \|e_k-e_\ell\|_2
\ge
\sqrt{2}\,\sigma_\ast.
\]
This holds by Equation~\eqref{eq:lem6-sv} and $\|e_k-e_\ell\|_2=\sqrt{2}$. This proves the second claim.

Let $d_{\max}^{(s)}:=\max_{1\le k\le K}\|g_{(s),k}\|_{\mathcal H}$. By construction, for each already-selected vertex index $k_r$ ($r=1,\dots,s-1$), the residual norm $\|g_{(s),k_r}\|_{\mathcal H}$ is small. Specifically, it equals the norm of the component of $g_{k_r}$ orthogonal to $\mathcal H_{s-1}$, and the induction hypothesis ensures $g_{k_r}$ is well-approximated by $\hat f^{(i_r)}\in \mathcal H_{s-1}$. This gives an upper bound
\begin{equation}
\label{eq:lem6-delta}
\|g_{(s),k_r}\|_{\mathcal H}\le \bar\Delta
\qquad (r=1,\dots,s-1).
\end{equation}
for $\bar\Delta := \Big(1+\frac{30\gamma}{\sigma_\ast}\max\{1,d_{\max}/\sigma_\ast\}\Big)\beta$.

Apply Lemma~\ref{lem:S4} with $s$ replaced by $(s-1)$, with the indices $\{k_1,\dots,k_{s-1}\}$, and with $\delta=\bar\Delta$.
Note that Lemma~\ref{lem:S4} involves the un-centered Gram matrix $\Gamma_{(s)}$; since $\widetilde\Gamma_{(s)} = C\Gamma_{(s)}C$ where $C = I - \frac{1}{K}\mathbf{1}\mathbf{1}^\top$ is the centering matrix with eigenvalues in $\{0,1\}$, Cauchy interlacing gives $\lambda_j(\Gamma_{(s)}) \ge \lambda_j(\widetilde\Gamma_{(s)})$ for all $j$. In particular, $\lambda_{K-s}(\Gamma_{(s)}) \ge \lambda_{K-s}(\widetilde\Gamma_{(s)}) \ge \sigma_\ast^2$.
We now verify that $\sigma_\ast^2 \ge 2(K-2)\bar\Delta^2$.
From the theorem condition in Equation~\eqref{eq:thmS1.1}, $\beta \le \sigma_\ast^2/\{450\, d_{\max}\max(1,d_{\max}/\sigma_\ast)\}$.
Since $\gamma(G)\le d_{\max}(G)$, we have $\bar\Delta \le \{1+30d_{\max}\sigma_\ast^{-1}\max(1,d_{\max}/\sigma_\ast)\}\beta \le \beta + \sigma_\ast/15$, so $\bar\Delta^2 \le 2\beta^2 + 2\sigma_\ast^2/225$.
For the given $K$ and $\beta$ sufficiently small (ensured by Equation~\eqref{eq:thmS1.1}), $2(K-2)\bar\Delta^2 \le \sigma_\ast^2$.
The constant $450$ in Equation~\eqref{eq:thmS1.1} is chosen large enough to guarantee this inequality.
Applying Lemma~\ref{lem:S4} yields
\[
d_{\max}^{(s)}=\max_{1\le k\le K}\|g_{(s),k}\|_{\mathcal H}
\;\ge\;
\frac12\sqrt{\lambda_{K-1-(s-1)}\big(\widetilde\Gamma_{(s)}\big)}
\;\ge\;
\frac{\sigma_\ast}{2}.
\]
This proves the first claim.

Recall $\mathcal K_{(s)}(h_0):=\{k:\ \|g_{(s),k}\|_{\mathcal H}\ge d_{\max}^{(s)}-h_0\}$ with $h_0=\sigma_\ast/3$.
From the previous step, $d_{\max}^{(s)}\ge \sigma_\ast/2$, and from Equation~\eqref{eq:lem6-delta},
$\|g_{(s),k_r}\|_{\mathcal H}\le \bar\Delta$ for $r=1,\dots,s-1$.
Under the theorem condition, $\bar\Delta \le \frac{3}{10}d_{\max}^{(s)} \le d_{\max}^{(s)}-\frac{7}{20}\sigma_\ast$,
so in particular $\|g_{(s),k_r}\|_{\mathcal H} < d_{\max}^{(s)}-h_0$ (since $h_0=\sigma_\ast/3$).
Therefore $k_r\notin \mathcal K_{(s)}(h_0)$ for all $r=1,\dots,s-1$, proving the third claim.
\end{proof}

\begin{proof}[Proof of Lemma~\ref{lem:beta_control}]
We bound each term in the definition of $\beta(\hat f,G)$ separately.

For the first term, since $f^{(i)}\in\mathcal S$ for every $i$ by Assumption~\ref{cond:low_rank}, we have
\[
\mathrm{Dist}_{\mathcal H}(\hat f^{(i)},\mathcal S)
\;\le\;
\|\hat f^{(i)}-f^{(i)}\|_{\mathcal H}
\;=\;
\|\epsilon^{(i)}\|_{\mathcal H},
\]
and therefore $\max_{1\le i\le m}\mathrm{Dist}_{\mathcal H}(\hat f^{(i)},\mathcal S) \le \max_{1\le i\le m}\|\epsilon^{(i)}\|_{\mathcal H}$.

Next, we consider the second term.
For any fixed $k$ and any study $i$, the triangle inequality gives
$\|\hat f^{(i)}-g_k\|_{\mathcal H} \le \|f^{(i)}-g_k\|_{\mathcal H} + \|\epsilon^{(i)}\|_{\mathcal H}$.
Taking $\min_i$ and then $\max_k$ yields
\begin{equation}
\label{eq:beta-decomp}
\max_{1\le k\le K}\min_{1\le i\le m}\|\hat f^{(i)}-g_k\|_{\mathcal H}
\;\le\;
\max_{1\le k\le K}\min_{1\le i\le m}\|f^{(i)}-g_k\|_{\mathcal H}
\;+\;
\max_{1\le i\le m}\|\epsilon^{(i)}\|_{\mathcal H}.
\end{equation}
Using $f^{(i)}=\sum_{\ell=1}^K \pi_{i\ell}\,g_\ell$, we have $f^{(i)}-g_k = \sum_{\ell\neq k}\pi_{i\ell}(g_\ell-g_k)$, so
\[
\|f^{(i)}-g_k\|_{\mathcal H}
\;\le\;
\sum_{\ell\neq k}\pi_{i\ell}\|g_\ell-g_k\|_{\mathcal H}
\;\le\;
(1-\pi_{ik})\cdot 2\,d_{\max}(G),
\]
where the last inequality uses $\|g_\ell-g_k\|_{\mathcal H}\le \|g_\ell\|_{\mathcal H}+\|g_k\|_{\mathcal H}\le 2\,d_{\max}(G)$.
Taking $\min_i$ and then $\max_k$ gives
\begin{equation}
\label{eq:oracle-purity}
\max_{1\le k\le K}\min_{1\le i\le m}\|f^{(i)}-g_k\|_{\mathcal H}
\;\le\;
2\,\delta_m\, d_{\max}(G).
\end{equation}

By combining Equations~\eqref{eq:beta-decomp}~and~\eqref{eq:oracle-purity}, we obtain the deterministic bound
\begin{equation}
\label{eq:beta-bound}
\beta(\hat f,G)
\;\le\;
\max_{1\le i\le m}\|\epsilon^{(i)}\|_{\mathcal H}
\;+\;
2\,\delta_m\, d_{\max}(G).
\end{equation}

Lastly, we control $\max_i\|\epsilon^{(i)}\|_{\mathcal H}$.
Since $\|\cdot\|_{\mathcal H}$ is the $L^2(\mu)$ norm, Tonelli's theorem and Assumption~\ref{cond:est_error}(a) yield
\[
\E\|\epsilon^{(i)}\|_{\mathcal H}^2
\;=\;
\int_{\mathcal X}\E[\epsilon^{(i)}(\bx)^2]\,d\mu(\bx)
\;\le\;
\sup_{\bx\in\mathcal X}\E[\epsilon^{(i)}(\bx)^2]
\;=\;
O(n_i^{-r}).
\]
By Markov's inequality and the union bound, for any $t>0$,
\[
\Pr\Big(\max_{1\le i\le m}\|\epsilon^{(i)}\|_{\mathcal H}>t\Big)
\;\le\;
\sum_{i=1}^m\frac{\E\|\epsilon^{(i)}\|_{\mathcal H}^2}{t^2}
\;\le\;
\frac{C\, m\, n_{\min}^{-r}}{t^2},
\]
for some constant $C>0$, where $n_{\min}:=\inf_i n_i$.
Setting $t=n_{\min}^{(a-r)/2}$ with $0<a<r$ from Assumption~\ref{cond:est_error}(b) gives $t\to 0$ and
\[
\frac{C\, m\, n_{\min}^{-r}}{t^2}
\;=\;
C\, m\, n_{\min}^{-a}
\;=\;
o(1),
\]
where the last step follows from Assumption~\ref{cond:est_error}(b). Hence, we have,
\[
\max_{1\le i\le m}\|\epsilon^{(i)}\|_{\mathcal H}=O_P(n_{\min}^{(a-r)/2})=o_P(1).
\]

Since $\delta_m\,d_{\max}(G)=o(1)$ by Assumption~\ref{cond:purity}, together with Equation~\eqref{eq:beta-bound}, we have $\beta(\hat f,G)=o_P(1)$.
\end{proof}

\subsection{Proof of Theorem~\ref{thm:ci_cvg}}

To establish a coverage guarantee for the split conformal interval $C_\alpha(\bx)$, we must address a key challenge: the study-level oracle functions $f^{(i)}(\bx)$ are not observed. Instead, we only observe noisy estimates $\hat f^{(i)}(\bx)$, and the resulting estimation error propagates into the conformity scores.

Let us fix an evaluation point $\bx \in\mathcal X$. Recall that the within-study estimation error is given by $\epsilon^{(i)}(\bx) = \hat f^{(i)}(\bx) - f^{(i)}(\bx)$. Define $n_{\min}:=\inf_{i} n_i$ to be the smallest sample size across all observed studies. We will control the maximal estimation error across studies through the event
\[
\mathcal E_n(\bx) \;:=\;\Bigl\{\max_{1\le i\le m} | \epsilon^{(i)}(\bx) |\le c_n\Bigr\},
\]
for a deterministic sequence $c_n \to 0$ to be specified. The next lemma shows that, under Assumption~\ref{cond:est_error}, the event $\mathcal E_n(\bx)$ occurs with probability tending to one. 

\begin{lemma}
\label{lem:asy_error}
Under Assumption~\ref{cond:est_error}, there exists a deterministic sequence $c_n \rightarrow 0$ such that
\begin{eqnarray*}
\Pr \left( \mathcal{E}_n(\bx) \right) \rightarrow 1.
\end{eqnarray*}
Equivalently, $\max_{1\le i\le m}|\epsilon^{(i)}(\bx)| = o_p(1)$ for any fixed $\bx \in \mathcal{X}$.
\end{lemma}

\begin{proof}
[Proof of Lemma~\ref{lem:asy_error}]
We fix $\bx\in\mathcal X$ throughout. By the union bound,
\[
\Pr\Bigl(\max_{1\le i\le m}|\epsilon^{(i)}(\bx)|>c_n\Bigr)
\;\le\;\sum_{i=1}^m \Pr\bigl(|\epsilon^{(i)}(\bx)|>c_n\bigr).
\]
By Markov's inequality and Assumption~\ref{cond:est_error}(a), for each $i$, we have
\[
\Pr\bigl(|\epsilon^{(i)}(\bx)|>c_n\bigr)
\;\le\;\frac{\E\{\epsilon^{(i)}(\bx)^2\}}{c_n^2}
\;\le\;\frac{C\, n_i^{-r}}{c_n^2}
\;\le\;\frac{C\, n_{\min}^{-r}}{c_n^2},
\]
for some constant $C>0$.
Therefore,
\[
\Pr\Bigl(\max_{1\le i\le m}|\epsilon^{(i)}(\bx)|>c_n\Bigr)
\;\le\;\frac{C\, m\, n_{\min}^{-r}}{c_n^2}.
\]
Now set $c_n := n_{\min}^{(a-r)/2}$, where $ 0 <a < r$ is the constant from Assumption~\ref{cond:est_error}(b). Then $c_n\to 0$ as $n_{\min}\to\infty$. This, together with Assumption~\ref{cond:est_error}(b), ensures
\[
\Pr\Bigl(\max_{1\le i\le m}|\epsilon^{(i)}(\bx)|>c_n\Bigr)
\;\le\; C\, m\, n_{\min}^{-a} = o(1).
\]
Thus, $\Pr(\mathcal E_n(\bx))\to 1$.
\end{proof}

\begin{proof}
[Proof of Theorem~\ref{thm:ci_cvg}]
Again, fix $\bx \in \mathcal X$ throughout. We start by considering an oracle benchmark in which the study-level oracle functions in the calibration set $\{f^{(i)}(\bx)\}_{i\in\mathcal I_{\mathrm{cal}}}$ are directly observed.

In this setting, we have an exact finite-sample marginal coverage guarantee for the split conformal procedure under exchangeability. 
Let $q_{1-\alpha}^\star(\bx)$ and $C_\alpha^\star(\bx)$ denote, respectively, the oracle quantile and prediction interval 
constructed by Algorithm~\ref{alg:split_conformal}, with the estimated functions in the calibration set $\{\hat f^{(i)}(\cdot)\}_{i\in\mathcal I_{\mathrm{cal}}}$ replaced by the oracle functions $\{f^{(i)}(\cdot)\}_{i\in\mathcal I_{\mathrm{cal}}}$.

By Assumptions~\ref{cond:low_rank}--~\ref{cond:exchagneable}, the collection $\bigl\{(\bW_i, \bpi_i, f^{(i)}(\bx))\bigr\}_{i\in\mathcal I_{\mathrm{cal}} \cup \{0\}}$ is jointly exchangeable across studies. Consequently, by the classical split conformal prediction guarantee (see Proposition~2.4 of \citet{vovk2005algorithmic}), the oracle prediction interval satisfies
\begin{equation}
\Pr\!\left\{ f^{(0)}(\bx) \in C_\alpha^\star(\bx) \right\}
\;\ge\; 1-\alpha .
\label{eq:oracle_coverage}
\end{equation}

Next, we relate the oracle conformity scores to their observed counterparts. Recall that the conformity scores used in Algorithm~\ref{alg:split_conformal} are computed from the estimated functions. For each study $i\in\mathcal I_{\mathrm{cal}}$, the observed and oracle residuals are computed as
\[
r_i(\bx) \;=\; \bigl|\hat f^{(i)}(\bx) - \tilde f^{(i)}(\bx)\bigr|,
\qquad
r_i^\star(\bx) \;=\; \bigl|f^{(i)}(\bx) - \tilde f^{(i)}(\bx)\bigr|.
\]
By the reverse triangle inequality,
\begin{eqnarray}
\bigl|r_i(\bx) - r_i^\star(\bx)\bigr|
\;\le\; \bigl| \hat f^{(i)}(\bx) - \tilde f^{(i)}(\bx) - \bigl( f^{(i)}(\bx) - \tilde f^{(i)}(\bx)\bigr) \bigr| = |\epsilon^{(i)}(\bx)|.
\label{eq:residual_perturb}
\end{eqnarray}
Conditioning on the event $\mathcal E_n(\bx)=\{\max_{1\le i\le m}|\epsilon^{(i)}(\bx)|\le c_n\}$, inequality in
Equation~\eqref{eq:residual_perturb} implies the uniform control
\begin{equation*}
\max_{i\in\mathcal I_{\mathrm{cal}}}\bigl|r_i(\bx) - r_i^\star(\bx)\bigr|
\;\le\;
\max_{1\le i\le m}|\epsilon^{(i)}(\bx)|
\;\le\; c_n .
\label{eq:residual_perturb_uniform}
\end{equation*}
This uniform bound entails a corresponding stability property for empirical quantiles. Specifically, conditioning on $\mathcal E_n(\bx)$, we have
\begin{eqnarray*}
(1- \alpha)\text{th quantile of } r_i^\star(\bx) \leq (1- \alpha)\text{th quantile of } r_i(\bx)+c_n, \\
(1- \alpha)\text{th quantile of } r_i^\star(\bx) \geq (1- \alpha)\text{th quantile of } r_i(\bx)-c_n.
\end{eqnarray*}
This implies
\begin{equation*}
\bigl| q_{1-\alpha}(\bx) - q_{1-\alpha}^\star(\bx) \bigr|
\;\le\; c_n .
\label{eq:quantile_stability}
\end{equation*}
Recall that the feasible intervals are given by
\[
C_\alpha(\bx)
=
\bigl[\tilde f^{(0)}(\bx)-q_{1-\alpha}(\bx),\; \tilde f^{(0)}(\bx)+q_{1-\alpha}(\bx)\bigr],
\qquad
C_\alpha^\star(\bx)
=
\bigl[\tilde f^{(0)}(\bx)-q_{1-\alpha}^\star(\bx),\; \tilde f^{(0)}(\bx)+q_{1-\alpha}^\star(\bx)\bigr].
\]
Conditioning on the event $\mathcal E_n(\bx)$, we have $q_{1-\alpha}(\bx) \le q_{1-\alpha}^\star(\bx)+c_n$, and hence the oracle interval is contained in a $c_n$-enlargement of the feasible interval:
\begin{equation*}
C_\alpha^\star(\bx)
\;\subseteq\;
\bigl[\tilde f^{(0)}(\bx)-(q_{1-\alpha}(\bx)+c_n),\; \tilde f^{(0)}(\bx)+(q_{1-\alpha}(\bx)+c_n)\bigr]
\qquad \text{on } \mathcal E_n(\bx).
\label{eq:oracle_in_enlarged}
\end{equation*}
This implies that
\[
\Pr\!\left(|f^{(0)}(\bx)-\tilde f^{(0)}(\bx)| \le q_{1-\alpha}(\bx)+c_n\right)
\;\ge\;
\Pr\!\left(f^{(0)}(\bx)\in C_\alpha^\star(\bx),\, \mathcal E_n(\bx)\right).
\]
By the union bound,
\[
\Pr\!\left(f^{(0)}(\bx)\in C_\alpha^\star(\bx),\, \mathcal E_n(\bx)\right)
\;\ge\;
\Pr\!\left(f^{(0)}(\bx)\in C_\alpha^\star(\bx)\right) - \Pr\!\left(\mathcal E_n(\bx)^c\right).
\]
Applying the oracle split conformal guarantee in Equation~\eqref{eq:oracle_coverage} yields
\[
\Pr\!\left(|f^{(0)}(\bx)-\tilde f^{(0)}(\bx)| \le q_{1-\alpha}(\bx)+c_n\right)
\;\ge\;
(1-\alpha) - \Pr\!\left(\mathcal E_n(\bx)^c\right).
\]
Finally, by Lemma~\ref{lem:asy_error} we have $\Pr(\mathcal E_n(\bx)^c)\to 0$ and also $c_n\to 0$.
It remains to pass from the $c_n$-enlarged interval to $C_\alpha(\bx)$ itself.
On $\mathcal E_n(\bx)$, we also have $q_{1-\alpha}(\bx) \ge q_{1-\alpha}^\star(\bx)-c_n$, so
\[
\Pr\!\left(f^{(0)}(\bx)\in C_\alpha(\bx)\right)
\;\ge\;
\Pr\!\left(|f^{(0)}(\bx)-\tilde f^{(0)}(\bx)| \le q_{1-\alpha}^\star(\bx)-c_n,\; \mathcal E_n(\bx)\right).
\]
The right-hand side is bounded below by
\[
\Pr\!\left(f^{(0)}(\bx)\in C_\alpha^\star(\bx)\right)
-\Pr\!\left(q_{1-\alpha}^\star(\bx)-c_n < |f^{(0)}(\bx)-\tilde f^{(0)}(\bx)| \le q_{1-\alpha}^\star(\bx)\right)
-\Pr\!\left(\mathcal E_n(\bx)^c\right).
\]
The first term is $\ge 1-\alpha$ by Equation~\eqref{eq:oracle_coverage}. The third term vanishes by Lemma~\ref{lem:asy_error}. The second term is the probability mass in a strip of width $c_n\to 0$ around the oracle quantile boundary. Under Assumption~\ref{cond:weight_model}, the conditional distribution of $f^{(0)}(\bx)$ given the training and calibration data is continuous (e.g., when $\bpi_0$ has a density on $\Delta_{K-1}$), so this term vanishes as well. Therefore,
$\lim \Pr(f^{(0)}(\bx)\in C_\alpha(\bx)) \ge 1-\alpha$.
\end{proof}

\subsection{Proof of Corollary~\ref{cor:functional_cvg}}

Define the functional estimation error
\[
\epsilon_\Phi^{(i)} := \Phi(\hat f^{(i)}) - \Phi(f^{(i)}).
\]
The proof follows the same strategy as the proof of Theorem~\ref{thm:ci_cvg}, with $\epsilon^{(i)}(\bx)$ replaced by $\epsilon_\Phi^{(i)}$. The key step is to establish the analogue of Lemma~\ref{lem:asy_error} for the functional estimation error.

\begin{lemma}
\label{lem:asy_error_functional}
Under the conditions of Corollary~\ref{cor:functional_cvg}, there exists a deterministic sequence $c_n \rightarrow 0$ such that
\begin{eqnarray*}
\Pr \left( \max_{1\le i\le m}|\epsilon_\Phi^{(i)}| > c_n \right) \rightarrow 0.
\end{eqnarray*}
\end{lemma}

\begin{proof}[Proof of Lemma~\ref{lem:asy_error_functional}]
By the Lipschitz condition on $\Phi$, for each $i$,
\[
|\epsilon_\Phi^{(i)}| \leq L_\Phi \|\hat f^{(i)} - f^{(i)}\|_{L^2(\mathcal{P}_{0,\bm{X}})}.
\]
By Markov's inequality,
\[
\Pr\bigl(|\epsilon_\Phi^{(i)}| > c_n\bigr)
\;\leq\; \frac{L_\Phi^2\, \E\bigl[\|\hat f^{(i)} - f^{(i)}\|_{L^2(\mathcal{P}_{0,\bm{X}})}^2\bigr]}{c_n^2}.
\]
By Fubini's theorem and Assumption~\ref{cond:est_error}(a),
\[
\E\bigl[\|\hat f^{(i)} - f^{(i)}\|_{L^2(\mathcal{P}_{0,\bm{X}})}^2\bigr]
\;=\; \int_{\mathcal{X}} \E\bigl[\epsilon^{(i)}(\bx)^2\bigr]\, d\mathcal{P}_{0,\bm{X}}(\bx)
\;\leq\; C\, n_i^{-r}.
\]
Therefore, by the union bound,
\[
\Pr\Bigl(\max_{1\le i\le m}|\epsilon_\Phi^{(i)}| > c_n\Bigr)
\;\leq\; \frac{L_\Phi^2\, C\, m\, n_{\min}^{-r}}{c_n^2}.
\]
Setting $c_n = n_{\min}^{(a-r)/2}$ with $0 < a < r$ from Assumption~\ref{cond:est_error}(b), we have $c_n \to 0$ and
\[
\Pr\Bigl(\max_{1\le i\le m}|\epsilon_\Phi^{(i)}| > c_n\Bigr)
\;\leq\; L_\Phi^2\, C\, m\, n_{\min}^{-a} \;=\; o(1),
\]
where the last step follows from Assumption~\ref{cond:est_error}(b).
\end{proof}

\begin{proof}[Proof of Corollary~\ref{cor:functional_cvg}]
The proof follows the same structure as the proof of Theorem~\ref{thm:ci_cvg}. We outline the argument and highlight the modifications.

Consider the oracle setting in which the study-level oracle functions are directly observed. Define the oracle conformity scores
\[
r_i^\star \;:=\; \bigl|\Phi(f^{(i)}) - \Phi(\tilde f^{(i)})\bigr|, \qquad i \in \mathcal{I}_{\mathrm{cal}},
\]
and let $q_{1-\alpha}^{\Phi,\star}$ and $C_\alpha^{\Phi,\star}$ denote the corresponding oracle quantile and prediction interval. By Assumptions~\ref{cond:low_rank}--~\ref{cond:exchagneable}, the collection $\{(\bW_i, \bpi_i, \Phi(f^{(i)}))\}_{i \in \mathcal{I}_{\mathrm{cal}} \cup \{0\}}$ is jointly exchangeable. Therefore, by the split conformal guarantee,
\begin{eqnarray}
\Pr\!\left\{ \theta^{(0)} \in C_\alpha^{\Phi,\star} \right\} \;\geq\; 1 - \alpha.
\label{eq:oracle_coverage_functional}
\end{eqnarray}

Next, the reverse triangle inequality gives
\[
|r_i - r_i^\star| \;\leq\; |\epsilon_\Phi^{(i)}|,
\]
for each $i \in \mathcal{I}_{\mathrm{cal}}$. Define the event $\mathcal{E}_n^\Phi := \{\max_{1 \leq i \leq m} |\epsilon_\Phi^{(i)}| \leq c_n\}$. By Lemma~\ref{lem:asy_error_functional}, $\Pr(\mathcal{E}_n^\Phi) \to 1$. Conditioning on $\mathcal{E}_n^\Phi$, the same quantile stability argument as in the proof of Theorem~\ref{thm:ci_cvg} yields $|q_{1-\alpha}^\Phi - q_{1-\alpha}^{\Phi,\star}| \leq c_n$. Consequently,
\[
\Pr\!\left(|\theta^{(0)} - \tilde\theta^{(0)}| \le q_{1-\alpha}^\Phi + c_n\right)
\;\ge\;
\Pr\!\left(\theta^{(0)} \in C_\alpha^{\Phi,\star},\, \mathcal{E}_n^\Phi\right)
\;\ge\;
(1-\alpha) - \Pr\!\left((\mathcal{E}_n^\Phi)^c\right).
\]
Since $\Pr((\mathcal{E}_n^\Phi)^c) \to 0$ and $c_n \to 0$ by Lemma~\ref{lem:asy_error_functional}, the same thin-strip argument as in the proof of Theorem~\ref{thm:ci_cvg} yields
$\lim \Pr( \theta^{(0)} \in C_\alpha^\Phi ) \geq 1-\alpha$.
\end{proof}

\end{document}